\newif\ifreview
\reviewfalse

\newif\ifarxiv
\arxivtrue

\ifreview
  \documentclass[journal,11pt,letterpaper,twopage,onecolumn]{IEEEtran}
\else
  \documentclass[journal,11pt,letterpaper,twopage]{IEEEtran}
\fi
\IEEEoverridecommandlockouts

\usepackage[utf8]{inputenc}
\usepackage[T1]{fontenc}
\usepackage{url}
\usepackage{ifthen}
\usepackage{cite}
\usepackage{balance}
\usepackage[cmex10]{amsmath} %
\usepackage{multirow}
\usepackage{tabularx}

\usepackage[dvipsnames,table]{xcolor}

\definecolor{tableaublue}{HTML}{1F77B4}
\definecolor{tableauorange}{HTML}{FF7F0E}
\definecolor{tableaugreen}{HTML}{2CA02C}
\definecolor{tableaured}{HTML}{D62728}
\definecolor{tableaupurple}{HTML}{9467BD}
\definecolor{tableaubrown}{HTML}{8C564B}
\definecolor{tableaupink}{HTML}{E377C2}
\definecolor{tableaugrey}{HTML}{7F7F7F}
\definecolor{tableaulightgreen}{HTML}{BCBD22}
\definecolor{tableaucyan}{HTML}{17BECF}

\usepackage{float}
\usepackage{latexsym}
\usepackage{amsthm}
\usepackage{thmtools}
\usepackage{thm-restate}
\usepackage{empheq}
\usepackage{bm}
\usepackage{booktabs}
\usepackage{subcaption}
\usepackage{tikz,pgfplots,cite}
\usetikzlibrary{matrix}

\usepackage{cite}
\usepackage{amsmath,amssymb,amsfonts}
\usepackage{algorithmic}
\usepackage{graphicx}
\usepackage{textcomp}

\usepackage{latexsym}
\usepackage{amsthm}
\usepackage{empheq}
\usepackage{bm}
\usepackage{booktabs}
\usepackage[inline]{enumitem}
\usepackage{array}
\usepackage{tkz-graph}
\usepackage{xspace}

\usepackage{pgfplots}
\pgfplotsset{compat=newest}
\pgfplotsset{every axis legend/.append style={%
cells={anchor=west}}
}
\usetikzlibrary{arrows}
\tikzset{>=stealth'}
\usetikzlibrary{intersections}

\usetikzlibrary{arrows.meta}
\usetikzlibrary{backgrounds}
\usetikzlibrary{graphs,graphs.standard}
\usepgfplotslibrary{patchplots}
\usepgfplotslibrary{fillbetween}
\pgfplotsset{%
    layers/standard/.define layer set={%
        background,axis background,axis grid,axis ticks,axis lines,axis tick labels,pre main,main,axis descriptions,axis foreground%
    }{
        grid style={/pgfplots/on layer=axis grid},%
        tick style={/pgfplots/on layer=axis ticks},%
        axis line style={/pgfplots/on layer=axis lines},%
        label style={/pgfplots/on layer=axis descriptions},%
        legend style={/pgfplots/on layer=axis descriptions},%
        title style={/pgfplots/on layer=axis descriptions},%
        colorbar style={/pgfplots/on layer=axis descriptions},%
        ticklabel style={/pgfplots/on layer=axis tick labels},%
        axis background@ style={/pgfplots/on layer=axis background},%
        3d box foreground style={/pgfplots/on layer=axis foreground},%
    },
}

\definecolor{redish}{rgb}{0.9, 0.17, 0.31}
\definecolor{fuchs}{rgb}{0.57, 0.36, 0.51}
\usepackage[colorlinks=true,linkcolor=fuchs]{hyperref}

\usepackage[normalem]{ulem}

\theoremstyle{definition}
\newtheorem{theorem}{Theorem}
\newtheorem{corollary}{Corollary}
\newtheorem{claim}{Claim}

\newtheorem{definition}{Definition}
\newtheorem{example}{Example}

\newtheorem{notation}{Notation}

\newtheorem{observation}{Observation}

\newtheorem{lemma}{Lemma}

\newtheorem{construction}{Construction}

\usepackage{comment}

\newtheorem{problem}{Problem}

\newcommand{\mB}{\mathcal{B}}
\newcommand{\mG}{\mathcal{G}}
\newcommand{\mV}{\mathcal{V}}
\newcommand{\mE}{\mathcal{E}}
\newcommand{\mA}{\mathcal{A}}
\newcommand{\mS}{\mathcal{S}}
\newcommand{\mW}{\mathcal{W}}
\newcommand{\mN}{\mathcal{N}}

\newcommand{\ecred}{red\xspace}
\newcommand{\ecgreen}{green\xspace}
\newcommand{\ecblue}{blue\xspace}
\newcommand{\ecblack}{black\xspace}
\newcommand{\ecpurple}{purple\xspace}

\newcommand{\lablen}{\ensuremath{m}\xspace}
\newcommand{\dbdim}{\ensuremath{d}\xspace}

\newcommand{\extend}{\operatorname{extend}}

\begin{document}
    \title{Achieving DNA Labeling Capacity with Minimum Labels through Extremal de Bruijn Subgraphs}

\ifarxiv 
  \markboth{Submitted to the IEEE Trans. Inform. Theory., May 2025}
\fi

\author{
    \IEEEauthorblockN{
    Christoph~Hofmeister, Anina~Gruica, Dganit~Hanania, Rawad~Bitar and Eitan~Yaakobi}
    \thanks{Part of this work was presented at the 2024 IEEE International Symposium on Information Theory (ISIT)~\cite{isitversion}.}
    \thanks{CH and RB are with the School of Computation, Information and Technology at the Technical University of Munich, Germany. 
    Emails: \{christoph.hofmeister, rawad.bitar\}@tum.de}
    \thanks{AG is with the Department of Applied Mathematics and Computer
Science, Technical University of Denmark, Denmark. Email: anigr@dtu.dk}
    \thanks{DH and EY are with the CS department of Technion---Israel Institute of Technology, Israel. Emails: \{dganit, yaakobi\}@cs.technion.ac.il} 
    \thanks{This project received funding by the German Research Foundation (DFG) under Grant Agreement Nos. BI 2492/5-1, BI 2492/1-1, and WA 3907/7-1. AG received funding from the Dutch Research Council through grant OCENW.KLEIN.539 and from the Villum Fonden through grant VIL”52303”.
    The research was also funded by the European Union (ERC, DNAStorage, 101045114 and EIC, DiDAX 101115134). Views and opinions expressed are however those of the authors only and do not necessarily reflect those of the European Union or the European Research Council Executive Agency. Neither the European Union nor the granting authority can be held responsible for them.
    }
}

\maketitle

\begin{abstract}
    \emph{DNA labeling} is a tool in molecular biology and biotechnology to visualize, detect, and study DNA at the molecular level. In this process, a DNA molecule is \emph{labeled} by a set of specific patterns, referred to as \emph{labels}, and is then imaged. The resulting image is modeled as an $(\ell+1)$-ary sequence, where $\ell$ is the number of labels, in which any non-zero symbol indicates the appearance of the corresponding label in the DNA molecule. The \emph{labeling capacity} refers to the maximum information rate that can be achieved by the labeling process for any given set of labels. The main goal of this paper is to study the minimum number of labels of the same length required to achieve the maximum labeling capacity of 2 for DNA sequences or $\log_2q$ for an arbitrary alphabet of size $q$. The solution to this problem requires the study of path unique subgraphs of the de Bruijn graph with the largest number of edges. We provide upper and lower bounds on this value. We draw new connections to existing literature that let us prove an asymptotic result as the label length tends to infinity.
            
\end{abstract}

\section{Introduction}

Labeling DNA molecules with fluorescent markers provides a powerful means to investigate their structure and function at the molecular level. This technique supports a wide range of applications, particularly in genomics, microbiology, and molecular medicine~\cite{moter2000fluorescence, chen2018efficient, gruszka2021single}.
Labeling methods vary and include fluorescence in situ hybridization (FISH)\cite{moter2000fluorescence}, CRISPR-based systems\cite{chen2018efficient,ma2015multicolor}, and approaches involving methyltransferases\cite{deen2017methyltransferase}.
One notable application of the labeling method is the optical mapping process~\cite{levy2013beyond, muller2017optical}. In~\cite{Nogin2023}, a mathematical model was developed to optimize labels for specific applications within this process.

Understanding the information-theoretic limits of the DNA labeling process was first examined in~\cite{hanania2023capacity}. Consider a DNA sequence $\boldsymbol{x} \in \{A, C, G, T\}^n$ and let $\boldsymbol{\alpha}\in\{A,C,G,T\}^\lablen$ be a \emph{label} of length $\lablen$. The labeling process identifies all positions in $\boldsymbol{x}$ where the label $\boldsymbol{\alpha}$ appears. The result is a binary sequence $\boldsymbol{y}$ of length $n$, where $y_i = 1$ if and only if $(x_i,\ldots,x_{i+\lablen-1})= \boldsymbol{\alpha}$. 

\colorlet{examplecolor1}{tableaublue}
\colorlet{examplecolor2}{tableaured}
In general, this process can be extended for any $\ell>0$ labels. In case $\ell$ different labels are used, the output becomes a sequence over $\{0, 1, \dots, \ell\}$, where each non-zero value indicates a match with a particular label. For example, applying $\boldsymbol{\alpha}_1 = \textcolor{examplecolor2}{AT}$ and $\boldsymbol{\alpha}_2 = \textcolor{examplecolor1}{TG}$ to $\boldsymbol{x} = \textcolor{examplecolor2}{AT}C\textcolor{examplecolor1}{TG}A\textcolor{examplecolor2}{A}\textcolor{examplecolor1}{TG}A$ produces the ternary sequence $\boldsymbol{y} = (\textcolor{examplecolor2}{1}, 0, 0, \textcolor{examplecolor1}{2}, 0, 0, \textcolor{examplecolor2}{1}, \textcolor{examplecolor1}{2}, 0, 0)$.
The \emph{labeling capacity}, introduced in~\cite{hanania2023capacity}, is the logarithmic ratio between the number of possible outputs by labeling and the length of the sequence. This value expresses the maximum asymptotic information rate that can be achieved through DNA labeling. The labeling capacity was fully determined for nearly all types of single labels, as well as for several cases involving multiple labels.

Another problem presented in~\cite{hanania2023capacity} asks for the minimum number of labels of the same length needed to distinguish `almost' all DNA sequences by their labeling output sequence. Here, by `almost,' we refer to achieving the maximal labeling capacity of $\log_2(4)=2$ for DNA sequences. In~\cite{hanania2023capacity}, the cases where the length of the labels is~$1$ or $2$ were fully solved. 
For example, when the length of the labels is~1, the minimum number of labels to reach a labeling capacity of 2 is 3. However, for labels of size greater than~1, the problem becomes more challenging. Further, a simple extension of the results of \cite{hanania2023capacity, hanania2024capacityjournal} shows that the minimum number of required labels of length $\lablen>1$ necessary to achieve the maximum labeling capacity over an alphabet of size $q$, which is $\log_2(q)$, equals the minimum number of edges that need to be removed from the $(\lablen-1)$-dimensional de Bruijn graph over $q$ symbols to make it \emph{path unique}. A graph is called path unique, if for any $k \ge 1$, between any two vertices, there exists at most one walk of length $k$. The maximum number of edges in a path unique 1-dimensional de Bruijn graph over $q$ symbols is solved in~\cite{zhan2012Extremal}, and from this result, it was observed in~\cite{hanania2023capacity}, that the minimum number of length-$2$ labels to reach maximum capacity $\log_2(4)=2$ is 10. 

Although~\cite{hanania2023capacity,hanania2024capacityjournal} provide an exact characterization of the problem of reaching the largest capacity with the minimum number of labels, the problem of finding the largest subgraphs of the de Bruijn graph that are path unique remained open. The main goal of this paper is to study this question. 

The rest of this paper is organized as follows. In Section~\ref{sec:prelims}, we introduce the notation and formally define our problem. We denote by $s(q,\lablen)$ the minimum number of length-$\lablen$ labels over $\Sigma_q$ that are needed in order to achieve the maximum labeling capacity. To find this value, we study $\gamma(q,\dbdim)$, the maximum number of edges in a path unique subgraph of a \mbox{$\dbdim$-dimensional} de Bruijn graph over~$q$ symbols, which is related to $s(q,\lablen)$. In Section~\ref{sec:constructions}, we introduce constructions of path unique de Bruijn subgraphs, which provide two lower bounds on $\gamma(q,\dbdim)$. The first construction is for any $\dbdim$, and the second is restricted to $\dbdim=2$. In Section~\ref{sec:upperbound}, we prove an upper bound on the value of $\gamma(q,\dbdim)$.
We connect the problem of finding $\gamma(q, \dbdim)$ to problems in the literature and derive an asymptotic lower bound on $\gamma(q, \dbdim)$ for large $\dbdim$ in Section~\ref{sec:connections}.
We compare the different bounds in Section~\ref{sec:comparison}.
In Section~\ref{sec:conclusion}, we discuss open problems and conclude the paper. %

\section{Definitions and Preliminaries}\label{sec:prelims}

Throughout this paper, $q$, $\lablen$, $\dbdim$ and $n$ are positive integers, and ~$\Sigma_q$ denotes the $q$-ary alphabet $\{0,1,\dots,q-1\}$. 
For a positive integer $n$ we denote by $[n]$ the set $\{1,\dots,n\}$. For $\bm{x}=(x_1,\dots,x_n) \in \Sigma_q^n$ and $i,k \in [n]$ with $1 \le i \le n-k+1$ we let $\bm{x}_{[i;k]} = (x_i,\dots,x_{i+k-1})$ and a \emph{label} $\bm{\alpha}\in \Sigma_q^\lablen$, $\lablen<n$ is a (relatively short) sequence over $\Sigma_q$. For any matrix $\bm{A}$ of size $n \times \lablen$, we denote by $|\bm{A}|$ the sum of the entries of $\bm{A}$, i.e., $|\bm{A}| := \sum_{i\in [n], j \in [\lablen]} A_{i,j}$. Next, we formally define the labeling process that is studied in this paper.

\begin{definition}[The labeling model~\cite{hanania2023capacity}]
Let $\bm{\alpha}_1,\dots,\bm{\alpha}_{\ell} \in \Sigma_{q}^\lablen$ be labels of length~$\lablen$. Denote by $\mA$ the set $\{\bm{\alpha}_1,\dots,\bm{\alpha}_{\ell}\}$, where $\bm{\alpha}_1 \le \dots \le \bm{\alpha}_{\ell}$ are ordered lexicographically. 
\begin{enumerate}[label=(\roman*), itemsep=0em]
    \item The \emph{$\mA$-labeling sequence} of $\bm{x}=(x_1,\dots,x_n) \in \Sigma_q^n$ is the sequence $L_{\mA}(\bm{x})=(c_1,\dots,c_n)\in \Sigma_{\ell+1}^n$, in which $c_i=j$ if $\bm{x}_{[i;\lablen]} = \bm{\alpha}_j$ and $i \le n-\lablen+1$. If such a~$j$ does not exist or $i \in \{n-\lablen,\dots,n\}$ then we set $c_i=0$.
    \item The \emph{labeling capacity} of $\mA$ is 
    \begin{align*}
        \mathsf{cap}(\mA) = \limsup_{n \to \infty} \frac{\log_2(|\{L_{\mA}(\bm{x}) : \bm{x} \in \Sigma_{q}^n\}|)}{n}.
    \end{align*}
\end{enumerate}
\end{definition}

\colorlet{examplecolor1}{tableaublue}
\colorlet{examplecolor2}{tableaured}
\begin{example}
Let $q=4$, $\lablen=2$ and $\mA=\{\bm{\alpha}_1,\bm{\alpha}_2\}$, where $\bm{\alpha}_1=\textcolor{examplecolor1}{(1,0)}$, $\bm{\alpha}_2=\textcolor{examplecolor2}{(2,2)}$. For $\bm{x}=(3,\textcolor{examplecolor1}{1,0},3,\textcolor{examplecolor2}{2,2,2},3,\textcolor{examplecolor1}{1,0})$ we have $L_\mA(\bm{x}) = (0,\textcolor{examplecolor1}{1},0,0,\textcolor{examplecolor2}{2,2},0,0,\textcolor{examplecolor1}{1},0)$.
\end{example}

In~\cite{hanania2023capacity, hanania2024capacityjournal}, the labeling capacity of a single label was computed for almost all label types (e.g., periodic labels), as well as for several cases involving multiple labels (e.g., overlapping labels).
In~\cite{hanania2023capacity, hanania2024capacityjournal}, the labeling capacity of a single label was computed for almost all cases and for several more cases of multiple labels. Another interesting problem presented in~\cite{hanania2023capacity, hanania2024capacityjournal} asks for the minimum number of labels of the same length that are needed in order to reach the largest labeling capacity of $\log_2q$. Formally, this problem, which is the main focus of this paper, is stated as follows. 
\begin{problem} \label{prob:motivation}
    Find the minimum number $\ell$ of length-$\lablen$ labels over $\Sigma_q$ such that there exists a set of labels $\mA=\{\bm{\alpha}_1,\dots,\bm{\alpha}_{\ell}\}$ achieving the largest capacity. More formally, find the value
    \begin{align*}
        s(q,\lablen) := \min\{\ell : \exists \mA \subseteq \Sigma_q^\lablen, |\mA| = \ell, \mathsf{cap}(\mA) = \log_2(q)\}.
    \end{align*}
\end{problem}

It is easy to see that for all $q$, $s(q,1) = q-1$. Moreover, it was proved in~\cite{hanania2023capacity} that $s(4,2) = 10$. This result was derived using a connection between Problem~\ref{prob:motivation} and a problem from graph theory, which is also the approach we take in this paper.  

To explain this connection and reformulate Problem~\ref{prob:motivation} in a graph-theoretic manner, we first define the well-studied objects on which this connection is based; see~\cite{de1946combinatorial} for a reference.

\begin{definition}[de Bruijn graphs] \label{def:debr}
The $\dbdim$-dimensional \emph{de Bruijn graph} $\mB_{q,\dbdim}=(\mV_{q,\dbdim},\mE_{q,\dbdim})$ over $\Sigma_q$ is the directed graph with $\mV_{q,\dbdim} = \Sigma_q^\dbdim$ and for $\bm{x}=(x_1,\dots,x_{\dbdim}) \in \Sigma_q^\dbdim$ and $\bm{y}=(y_1,\dots, y_{\dbdim}) \in \Sigma_q^\dbdim$ we have $(\bm{x},\bm{y}) \in \mE_{q,\dbdim}$ if and only if $(x_2,\dots,x_{\dbdim})=(y_1,\dots,y_{\dbdim-1})$, i.e., the last $\dbdim-1$ entries of~$\bm{x}$ coincide with the first $\dbdim-1$ entries of~$\bm{y}$.
\end{definition}

Note that for a positive integer $k$, a \emph{walk} of length $k$ (i.e., a sequence of vertices and connecting edges with a total of $k$ edges) in the de Bruijn graph $\mB_{q,\dbdim}$ can be seen as a sequence $\bm{s}=(s_1,\dots,s_{\dbdim+k})\in \Sigma_q^{\dbdim+k}$. In particular, an edge in $\mB_{q,\dbdim}$ corresponds to a sequence in $\Sigma_q^{\dbdim+1}$. As shown later in Theorem~\ref{thm:dunno}, Problem~\ref{prob:motivation} can be reformulated into a problem concerned with graphs satisfying a specific constraint:

\begin{definition}[Path unique graph]
Let $\mG=(\mV,\mE)$ be a directed graph. It is said that $\mG$ is \emph{path unique} if for any positive integer $k$ there exists at most one walk of length~$k$ between any two vertices.
\end{definition}

An \textit{edge-induced subgraph} is a subset of the edges of a graph together with any vertices that are their start- or endpoints. 
From now on, by a subgraph we mean an edge-induced subgraph. The following theorem shows the connection between path unique subgraphs of $\mB_{q,\dbdim}$ and the value of $s(q,\dbdim+1)$. Since it is a generalization of~\cite[Theorem 12]{hanania2024capacityjournal} and it can be proven analogously, we omit the proof here.

\begin{theorem} \label{thm:dunno}
    Let $\mG=(\mV,\mE)$ be a subgraph of $\mB_{q,\dbdim}$ and let $\mA=\Sigma_q^{\dbdim+1} \setminus \mE$. It holds that $\mathsf{cap}(\mA) = \log_2(q)$ if and only if~$\mG$ is path unique.
\end{theorem}

Intuitively, the walks on $\mG$ correspond one-to-one to the sequences $\bm{x}$ for which $L_\mathcal{A}(\bm{x})$ contains only zeros. If every such sequence can always be uniquely identified by the last label that occurred before it, the first label occurring after it, and its length, then almost all sequences can be distinguished based on their label sequence.

From Theorem~\ref{thm:dunno} it follows that finding the size of the largest path unique subgraph of $\mB_{q,\dbdim}$ is equivalent to finding the value of $s(q,\dbdim+1)$. This motivates the following reformulation of Problem~\ref{prob:motivation}. 

\begin{problem} \label{prob:gamma} 
For any $q$ and $\dbdim$, let $\Gamma(q,\dbdim)$ be the set of subgraphs of $\mB_{q,\dbdim}$ that are path unique. Find the value
\begin{align*}
    \gamma(q,\dbdim) := \max\{|\mE| : \mG=(\mV,\mE) \in \Gamma(q,\dbdim)\}.
\end{align*}
We call graphs in $\Gamma(q,\dbdim)$ that have $\gamma(q,\dbdim)$ edges \emph{optimal} subgraphs of $\mB_{q,\dbdim}$. 
\end{problem}

An immediate consequence of Theorem~\ref{thm:dunno} is the following.
\begin{corollary} \label{cor:relation}
We have $s(q,{\dbdim}+1)=q^{\dbdim+1}-\gamma(q,\dbdim)$.
\end{corollary}

Due to Corollary~\ref{cor:relation}, it is enough to study the value of~$\gamma(q,\dbdim)$ when attempting to solve Problem~\ref{prob:motivation} (and Problem~\ref{prob:gamma}) for labels of length $\lablen=\dbdim+1$, which will be the focus of the rest of the paper.

The problem of finding the value of $\gamma(q,1)$ was fully solved in~\cite{zhan2012Extremal}, and it concerns the complete graph of $q$ vertices. For example, it holds that $\gamma(4,1) = 6$ and a corresponding graph using the DNA alphabet is visualized in Fig.~\ref{fig:s41}.
\begin{theorem}[\textnormal{see~\cite[Theorem 1]{zhan2012Extremal}}]
The following holds:
\begin{align*}
    \gamma(q,1) = \begin{cases}
        \frac{(q+1)^2}{4} \quad &\textnormal{if $q$ is odd,} \\
        \frac{q(q+2)}{4} \quad &\textnormal{if $q$ is even.}
    \end{cases}
\end{align*}
\label{thm:extremal}
\end{theorem}

\begin{figure}[t]
  \centering
  \resizebox{0.35\textwidth}{!}{
    \begin{tikzpicture}[->,>=stealth',shorten >=1pt,thick]
      \SetGraphUnit{2} 
      \tikzset{VertexStyle/.style = {draw,circle,thick,minimum size=10mm, fill=black!10!white},thick} 
      \Vertices{line}{A,C,G,T}

      \Loop[dist=1cm,dir=NO](A)
      \Loop[dist=1cm,dir=NO](G)
      \draw[->] (A) to [right] node [above] {} (C);
      \draw[->] (A) to [bend right] node [above] {} (T);
      \draw[->] (G) to [right] node [above] {} (C);
      \draw[->] (G) to [right] node [above] {} (T);
  \end{tikzpicture}}
  \caption{\small An optimal path unique graph over $\Sigma_4=\{A,C,G,T\}$.}
  \label{fig:s41}
\end{figure} 

The approach of~\cite{zhan2012Extremal} relies on the following lemma.

\begin{lemma} \label{lem:adj}

  Let $\mG$ be a graph on $n$ vertices, and let $\bm{A}$ be its adjacency matrix. For any $i, j \in [n]$, and for any positive integer $k$ the entry in position $(i,j)$ of $\bm{A}^k$ equals the number of walks in $\mG$ starting at vertex $i$ and ending at vertex $j$. Further,~$\mG$ is a path unique graph if and only if for any positive integer~$k$, the entries of $\bm{A}^k$ are in $\{0,1\}$.
\end{lemma}

We denote the adjacency matrix of a de Bruijn graph $\mB_{q,\dbdim}$ by $\bm{B}_{q,\dbdim}$ (with vertices in lexicographic order).

\ifreview
  \section{Constructions of Path Unique Subgraphs of $\mB_{q,\dbdim}$}\label{sec:constructions}
\else
  \section{Constructions of Path Unique\\ Subgraphs of $\mB_{q,\dbdim}$}\label{sec:constructions}
\fi

In this section, we give two constructions of path unique subgraphs of $\mB_{q,\dbdim}$, which provide lower bounds on $\gamma(q,\dbdim)$. The first construction is applicable for any $q$ and $\dbdim$, while the second works for any $q$ and $\dbdim=2$ but gives a tighter lower bound.

\begin{construction} \label{constr:1}
Let $\mG^1_{q,\dbdim}=(\mV_{q,\dbdim}^1,\mE_{q,\dbdim}^1)$ be the subgraph of $\mB_{q,\dbdim}=(\mV_{q,\dbdim}, \mE_{q, \dbdim})$ with $\mV_{q,\dbdim}^1 = \mV_{q,\dbdim}$ and where the edge $(x_1,\dots,x_{\dbdim+1}) \in \mE_{q, \dbdim}$ is in $\mE_{q,\dbdim}^1$ if and only if
\begin{enumerate}[label=\Roman*), itemsep=0em]
    \item \label{const1-cond1}  $x_1 \le x_2$,
    \item \label{const1-cond2} it does \textbf{not} hold that $x_1 \le x_2 \le \dots \le x_{\dbdim+1} \le q-2$.
\end{enumerate}
\end{construction}

\begin{lemma} \label{lem:constr1count}
The total number of edges in $\mG^1_{q,\dbdim}$ is
\begin{align*}
    \left(\frac{q+1}{2q}\right)q^{\dbdim+1}- \binom{\dbdim+q-1}{\dbdim+1}.
\end{align*}
\end{lemma}
\begin{proof}
The proportion of edges $(x_1,\dots,x_{\dbdim+1}) \in \Sigma_q^{\dbdim+1}= \mE_{q, \dbdim}$ satisfying $x_1 \le x_2$ within the set of all edges is $(q+1)/(2q)$. The edges $(x_1,\dots,x_{\dbdim+1}) \in \Sigma_q^{\dbdim+1}$ for which $x_1 \le x_2 \le \dots \le x_{\dbdim+1} \le q-2$ are clearly a subset of the edges described in~\ref{const1-cond1} of Construction~\ref{constr:1}. There are a total of $\binom{\dbdim+q-1}{\dbdim+1}$ such edges (see e.g.~\cite[Chapter II.5]{feller1991introduction}) from which the lemma follows.
\end{proof}

\begin{example}
When constructing the graph $\mG^1_{3,3}$ we start with the graph shown in Fig.~\ref{fig:partgraph} consisting of 54 edges (where $(i,\bm{x}) \in \Sigma_q^{\dbdim}$ is the set of vertices whose first entry is equal to $i \in \Sigma_q$).
We then remove the edges in the set $$\{(0,0,0,0), (0,0,0,1),(0,0,1,1),(0,1,1,1),(1,1,1,1)\}$$ resulting in a subgraph of $\mB_{3,3}$ with $49$ edges, which is $\mG^1_{3,3}$.
\begin{figure}[t]
  \centering\resizebox{0.35\textwidth}{!}{
  \begin{tikzpicture}[->,>=stealth',shorten >=1pt,thick,round/.style={draw,circle,thick,minimum size=5mm,thick},]
\SetGraphUnit{2} 
\node[round, fill=black!10!white] at (0,0) (A) [] {$(0,\bm{x})$};
\node[round, fill=black!10!white] (B) at (3,0) [] {$(1,\bm{x})$};
\node[round, fill=black!10!white] (C) at (6,0) [] {$(2,\bm{x})$};

\Loop[dist=1cm,dir=SO](A)
\Loop[dist=1cm,dir=SO](B)
\Loop[dist=1cm,dir=SO](C)
\draw[->] (A) to [left] node [above] {} (B);
\draw[->] (A) to [bend left] node [above] {} (C);
\draw[->] (B) to [left] node [above] {} (C);
\end{tikzpicture}}
\caption{\small The edges of part~\ref{const1-cond1} of Construction~\ref{constr:1} for $q=3$.}
\label{fig:partgraph}
\end{figure} 
\end{example}

\begin{theorem} \label{prop:genlb}
The graph $\mG^1_{q,\dbdim}$ is path unique. In particular,
\begin{align*}
    \gamma(q,\dbdim) \ge \left(\frac{q+1}{2q}\right) q^{\dbdim+1}- \binom{\dbdim+q-1}{\dbdim+1}.
\end{align*}  
\end{theorem}
\begin{proof}

Let $\bm{s}=(s_1,\dots,s_k) \in \Sigma_q^k$ be a walk on $\mG^1_{q,\dbdim}$ from $\bm{x}=(x_1,\dots,x_{\dbdim})$ to $\bm{y}=(y_1,\dots,y_{\dbdim})$. If $k \le 2\dbdim$, then there is nothing to show. If $k > 2\dbdim$, we show that $s_{\dbdim+1}=s_{\dbdim+2}=\dots=s_{k-\dbdim}=q-1$, hence the walk is unique. We split the proof into two parts, and we prove each part separately.
\begin{claim} \label{claim:constr1_1}
We have $x_1 \le \dots \le x_{\dbdim}$.
\end{claim}
\begin{proof}[Proof of Claim~\ref{claim:constr1_1}]
If there exists an $\ell \in \{2,\dots,\dbdim\}$ with $x_{\ell} < x_{\ell-1}$, this would mean that there is an edge $(x_{\ell-1}, x_{\ell},\bm{z})$ in $\mG_{q,\dbdim}^1$ for some $\bm{z} \in \Sigma_q^{\dbdim-1}$ with $x_{\ell} < x_{\ell-1}$, which by~\ref{const1-cond1} of Construction~\ref{constr:1} is a contradiction.
\end{proof}
\begin{claim} \label{claim:constr1_2}
We have $s_{\dbdim+1}=q-1$.
\end{claim}
\begin{proof}[Proof of Claim~\ref{claim:constr1_2}]
We show that the only outgoing edge from~$\bm{x}$ is to $(x_2,\dots,q-1)$. Towards a contradiction, assume that there is an edge $(x_1,\dots,x_{\dbdim},i) \in \mE_{q,\dbdim}^1$ with $i \le q-2$. First, note that $x_{\dbdim} \le i$, as otherwise, similar to the proof of Claim~\ref{claim:constr1_1},  there is an edge $(x_{\dbdim},i, \bm{z})$ in $\mG_{q,\dbdim}^1$ for some $\bm{z} \in \Sigma_q^{\dbdim-1}$ with $x_{\dbdim} > i$, contradicting~\ref{const1-cond1} of Construction~\ref{constr:1}. Combined with Claim~\ref{claim:constr1_1}, it follows that $x_1 \le \dots \le x_{\dbdim} \le i \le q-2$ contradicting~\ref{const1-cond2} of Construction~\ref{constr:1}.
\end{proof}
Similar to the proof of Claim~\ref{claim:constr1_2}, if $k-2\dbdim \ge 2$, by considering the walk going from $(x_2,\dots,x_{\dbdim},s_{\dbdim+1})$ to $\bm{y}$, we have that $s_{\dbdim+2}=q-1$. We can show inductively that $s_{\dbdim+1}=s_{\dbdim+2}=\dots=s_{k-\dbdim}=q-1$.
\end{proof}

Theorem~\ref{prop:genlb} gives the following asymptotic lower bounds.

\begin{corollary}
The following statements hold:
\begin{enumerate}[label=(\roman*), itemsep=0em]
    \item For every fixed $q$, we have $$\lim_{\dbdim\rightarrow \infty}\frac{\gamma(q,\dbdim)}{q^{\dbdim+1}} \ge \frac{q+1}{2q}.$$
    \item For every fixed $\dbdim$, we have
    $$\lim_{q\rightarrow \infty}\frac{\gamma(q,\dbdim)}{q^{\dbdim+1}} \ge \frac{1}{2} - \frac{1}{(\dbdim+1)!}.$$
\end{enumerate}

\end{corollary}

\begin{figure*}[t]
\begin{center}
     \definecolor{tableaublue}{HTML}{1F77B4}
\definecolor{tableauorange}{HTML}{FF7F0E}
\definecolor{tableaugreen}{HTML}{2CA02C}
\definecolor{tableaured}{HTML}{D62728}
\definecolor{tableaupurple}{HTML}{9467BD}
\definecolor{tableaubrown}{HTML}{8C564B}
\definecolor{tableaupink}{HTML}{E377C2}
\definecolor{tableaugrey}{HTML}{7F7F7F}
\definecolor{tableaulightgreen}{HTML}{BCBD22}
\definecolor{tableaucyan}{HTML}{17BECF}

\begin{tikzpicture}[->,>=stealth',shorten >=1pt,thick,round/.style={draw,circle,thick,minimum size=5mm,thick},]
\SetGraphUnit{2}

\def\dx{1.2} %
\def\dy{1.6} %
\def\ox{0.75} %

\def\bc{0}
\node[round, fill=black!10!white] at (0*\dx, 0) (N1\bc) [] {$1\bc$};
\node[round, fill=black!10!white] at (1*\dx, 0) (N2\bc) [] {$2\bc$};
\node[round, fill=black!10!white] at (2*\dx, 0) (N3\bc) [] {$3\bc$};
\node[round, fill=black!10!white] at (3*\dx, 0) (N4\bc) [] {$4\bc$};

\node[round, fill=black!10!white] at (0*\dx-\ox,-\dy) (N\bc0) [] {$\bc0$};
\node[round, fill=black!10!white] at (1*\dx-\ox,-\dy) (N\bc4) [] {$\bc4$};
\node[round, fill=black!10!white] at (2*\dx-\ox,-\dy) (N\bc3) [] {$\bc3$};
\node[round, fill=black!10!white] at (3*\dx-\ox,-\dy) (N\bc2) [] {$\bc2$};
\node[round, fill=black!10!white] at (4*\dx-\ox,-\dy) (N\bc1) [] {$\bc1$};

\def\bdx{5*\dx} %
\def\bdy{0} %
\def\bc{1}
\node[round, fill=black!10!white] at (0*\dx+\bdx, -\bdy) (N4\bc) [] {$4\bc$};
\node[round, fill=black!10!white] at (1*\dx+\bdx, -\bdy) (N3\bc) [] {$3\bc$};
\node[round, fill=black!10!white] at (2*\dx+\bdx, -\bdy) (N2\bc) [] {$2\bc$};

\node[round, fill=black!10!white] at (0*\dx+\bdx-\ox,-\dy-\bdy) (N\bc1) [] {$\bc1$};
\node[round, fill=black!10!white] at (1*\dx+\bdx-\ox,-\dy-\bdy) (N\bc2) [] {$\bc2$};
\node[round, fill=black!10!white] at (2*\dx+\bdx-\ox,-\dy-\bdy) (N\bc3) [] {$\bc3$};
\node[round, fill=black!10!white] at (3*\dx+\bdx-\ox,-\dy-\bdy) (N\bc4) [] {$\bc4$};

\def\bdx{9*\dx} %
\def\bdy{0} %
\def\bc{2}
\node[round, fill=black!10!white] at (0*\dx+\bdx, -\bdy) (N4\bc) [] {$4\bc$};
\node[round, fill=black!10!white] at (1*\dx+\bdx, -\bdy) (N3\bc) [] {$3\bc$};

\node[round, fill=black!10!white] at (0*\dx+\bdx-\ox,-\dy-\bdy) (N\bc2) [] {$\bc2$};
\node[round, fill=black!10!white] at (1*\dx+\bdx-\ox,-\dy-\bdy) (N\bc3) [] {$\bc3$};
\node[round, fill=black!10!white] at (2*\dx+\bdx-\ox,-\dy-\bdy) (N\bc4) [] {$\bc4$};

\def\bdx{12*\dx} %
\def\bdy{0} %
\def\bc{3}
\node[round, fill=black!10!white] at (0*\dx+\bdx, -\bdy) (N4\bc) [] {$4\bc$};

\node[round, fill=black!10!white] at (0*\dx+\bdx-\ox,-\dy-\bdy) (N\bc3) [] {$\bc3$};
\node[round, fill=black!10!white] at (1*\dx+\bdx-\ox,-\dy-\bdy) (N\bc4) [] {$\bc4$};

\def\bdx{14*\dx} %
\def\bdy{0} %
\def\bc{4}
\node[round, fill=black!10!white] at (0*\dx+\bdx-\ox,-\dy-\bdy) (N\bc4) [] {$\bc4$};

\Loop[dist=0.8cm, dir=NOWE, color=tableaublue](N00)
\Loop[dist=0.8cm, dir=NOWE, color=tableaublue](N11)
\Loop[dist=0.8cm, dir=NOWE, color=tableaublue](N22)
\Loop[dist=0.8cm, dir=NOWE, color=tableaublue](N33)
\Loop[dist=0.8cm, dir=NOWE, color=tableaublue](N44)

\draw[->, color=black!80] (N10) to [left] node  {} (N00);
\draw[->, color=black!80] (N10) to [left] node  {} (N01);
\draw[->, color=black!80] (N10) to [left] node  {} (N02);
\draw[->, color=black!80] (N10) to [left] node  {} (N03);
\draw[->, color=black!80] (N10) to [left] node  {} (N04);
\draw[->, color=black!80] (N20) to [left] node  {} (N00);
\draw[->, color=black!80] (N20) to [left] node  {} (N01);
\draw[->, color=black!80] (N20) to [left] node  {} (N02);
\draw[->, color=black!80] (N20) to [left] node  {} (N03);
\draw[->, color=black!80] (N20) to [left] node  {} (N04);
\draw[->, color=black!80] (N30) to [left] node  {} (N00);
\draw[->, color=black!80] (N30) to [left] node  {} (N01);
\draw[->, color=black!80] (N30) to [left] node  {} (N02);
\draw[->, color=black!80] (N30) to [left] node  {} (N03);
\draw[->, color=black!80] (N30) to [left] node  {} (N04);
\draw[->, color=black!80] (N40) to [left] node  {} (N00);
\draw[->, color=black!80] (N40) to [left] node  {} (N01);
\draw[->, color=black!80] (N40) to [left] node  {} (N02);
\draw[->, color=black!80] (N40) to [left] node  {} (N03);
\draw[->, color=black!80] (N40) to [left] node  {} (N04);

\draw[->, color=black!80] (N21) to [left] node  {} (N11);
\draw[->, color=black!80] (N21) to [left] node  {} (N12);
\draw[->, color=black!80] (N21) to [left] node  {} (N13);
\draw[->, color=black!80] (N21) to [left] node  {} (N14);
\draw[->, color=black!80] (N31) to [left] node  {} (N11);
\draw[->, color=black!80] (N31) to [left] node  {} (N12);
\draw[->, color=black!80] (N31) to [left] node  {} (N13);
\draw[->, color=black!80] (N31) to [left] node  {} (N14);
\draw[->, color=black!80] (N41) to [left] node  {} (N11);
\draw[->, color=black!80] (N41) to [left] node  {} (N12);
\draw[->, color=black!80] (N41) to [left] node  {} (N13);
\draw[->, color=black!80] (N41) to [left] node  {} (N14);

\draw[->, color=black!80] (N32) to [left] node  {} (N22);
\draw[->, color=black!80] (N32) to [left] node  {} (N23);
\draw[->, color=black!80] (N32) to [left] node  {} (N24);
\draw[->, color=black!80] (N42) to [left] node  {} (N22);
\draw[->, color=black!80] (N42) to [left] node  {} (N23);
\draw[->, color=black!80] (N42) to [left] node  {} (N24);

\draw[->, color=black!80] (N43) to [left] node  {} (N33);
\draw[->, color=black!80] (N43) to [left] node  {} (N34);

\draw[->, color=tableaured] (N11) to [bend right=20] node {} (N12);
\draw[->, color=tableaured] (N11) to [bend right=30] node {} (N13);
\draw[->, color=tableaured] (N11) to [bend right=35] node {} (N14);

\draw[->, color=tableaured] (N22) to [bend right=20] node  {} (N23);
\draw[->, color=tableaured] (N22) to [bend right=30] node  {} (N24);

\draw[->, color=tableaured] (N33) to [bend right=20] node  {} (N34);

\draw[->, color=tableaugreen] (N01) to [bend right=30] node  {} (N11);
\draw[->, color=tableaugreen] (N01) to [bend right=35] node  {} (N12);
\draw[->, color=tableaugreen] (N01) to [bend right=38] node  {} (N13);
\draw[->, color=tableaugreen] (N01) to [bend right=40] node  {} (N14);

\draw[->, color=tableaugreen] (N02) to [bend right=35] node  {} (N22);
\draw[->, color=tableaugreen] (N02) to [bend right=38] node  {} (N23);
\draw[->, color=tableaugreen] (N02) to [bend right=40] node  {} (N24);

\draw[->, color=tableaugreen] (N03) to [bend right=38] node  {} (N33);
\draw[->, color=tableaugreen] (N03) to [bend right=40] node  {} (N34);

\draw[->, color=tableaugreen] (N04) to [bend right=40] node  {} (N44);

\draw[->, color=tableaupurple] (N42) to [bend right=20] node  {} (N21);
\draw[->, color=tableaupurple] (N43) to [bend right=25] node  {} (N31);
\draw[->, color=tableaupurple] (N43) to [bend right=20] node  {} (N32);

\end{tikzpicture}
\end{center}
\vspace{-2.5em}
\caption{\small The graph $\mG^2_{5}$ from Construction~\ref{constr:2}. For compactness, each vertex $(x_1, x_2) \in \Sigma_q^2$ is labeled by $x_1x_2$, e.g., $(0, 1)$ by $01$. Edges satisfying Condition~\ref{const2-cond1} are \textcolor{tableaublue}{\emph{blue}}, edges satisfying Condition~\ref{const2-cond2} are \textcolor{tableaured}{\emph{red}}, edges satisfying Condition~\ref{const2-cond3} are \emph{black}, edges satisfying Condition~\ref{const2-cond4} are \textcolor{tableaugreen}{\emph{green}}, and edges satisfying Condition~\ref{const2-cond5} are \textcolor{tableaupurple}{\emph{purple}}. }
\label{fig:bigdganit}
\end{figure*}

Next, we present our second construction. 
\begin{construction} \label{constr:2}
  Let $ \mG^2_{q}=(\mV^2_{q}, \mE^2_{q})$ be the subgraph of $\mB_{q,2}=(\mV_{q,2}, \mE_{q, 2})$ with $\mV^2_{q} = \mV_{q, 2}$ where the edge $(x_1, x_2, x_3) \in \mE_{q,2}$ is in $\mE^2_{q}$ if and only if it meets one of the following (mutually exclusive) conditions:\begin{enumerate}[label=\Roman*), itemsep=0em]
    \item \label{const2-cond1} $x_1=x_2=x_3$,
    \item \label{const2-cond2} $0<x_1=x_2<x_3$,
    \item \label{const2-cond3} $x_1>x_2$ and $x_2 \leq x_3$,
    \item \label{const2-cond4} $0=x_1<x_2\leq x_3$, or
    \item \label{const2-cond5} $q-1=x_1>x_2>x_3>0$.
  \end{enumerate}
\end{construction}

We partition the vertices in $\mV^2_{q, 2}$ into $q$ blocks indexed from~$0$ to $q-1$. 
Vertex $(x_1, x_2)$ is in block $\min(x_1,x_2)$.
Within each block, we divide the vertices into two parts. We say that $(x_1, x_2)$ is in the \emph{top part} if $x_1>x_2$ and in the \emph{bottom part} otherwise. For clarity, we assign a color to each edge $(x_1,x_2,x_3) \in \mE^2_{q}$ based on which of the above conditions it meets; cf. Fig.~\ref{fig:bigdganit} for a depiction.

More precisely:
\begin{enumerate}[label=(\roman*), itemsep=0em]
    \item Condition~\ref{const2-cond1} ({\em \ecblue}) is met by self-loops.%
    \item Condition~\ref{const2-cond2} ({\em \ecred}) is met by edges from the constant vertex $(x_1, x_1)$ in the bottom part of a block $x_1$, with $x_1>0$, to another vertex in the bottom part of block $x_1$. %
    \item Condition~\ref{const2-cond3} ({\em \ecblack}) is met by edges from the top part of a block to the bottom part of the same block. %
    \item Condition~\ref{const2-cond4} ({\em \ecgreen}) is met by edges from the bottom part in block $0$ to the bottom part of another block. %
    \item Finally, Condition~\ref{const2-cond5} ({\em \ecpurple}) is met by edges from the top part of a block $x_2$, with $1<x_2<q-1$, to the top part of another block $x_3$, with $0<x_3<x_2$. %
\end{enumerate}

\begin{lemma}
  The number of edges in $\mG_{q}^2$ is
  \begin{align*}
    \frac{1}{3} q^3 + \frac{3}{2} q^2 -\frac{23}{6} q + 4.
  \end{align*}
\end{lemma}
\begin{proof}
  We count the number of triples $(x_1, x_2, x_3) \in \Sigma^3_q$ that fulfill each condition.
  Condition \ref{const2-cond1} is met by $q$ triples, Condition \ref{const2-cond2} by ${q-1 \choose 2}$ triples, Condition~\ref{const2-cond3} by $\sum_{i=1}^{q-1} i^2 + i = 2\binom{q}{2}+2\binom{q}{3}$ triples, Condition~\ref{const2-cond4} by ${q \choose 2}$ triples, and Condition~\ref{const2-cond5} by $q-2 \choose 2$ triples.
  Simplifying gives the desired expression.
\end{proof}

\begin{example}
The graph $\mG^2_{5}$ from Construction~\ref{constr:2} is shown in Fig.~\ref{fig:bigdganit}.
  The color of each edge corresponds to the condition in Construction~\ref{constr:2} it fulfills.
  The vertices are partitioned according to the blocks of the construction, with the vertices in the top part arranged above the vertices of the bottom part.
  \label{ex:constr2}
\end{example}

\begin{theorem}\label{prop:lbm2}
The graph $\mG^2_{q}$ is path unique. In particular, we have
\begin{align*}
    \gamma(q,2) \ge  \frac{1}{3} q^3 + \frac{3}{2} q^2 -\frac{23}{6} q + 4.
\end{align*}
\end{theorem}
Before we prove the theorem, we compare the bounds obtained from Theorem~\ref{prop:genlb} and Theorem~\ref{prop:lbm2}.

\begin{observation} \label{obs:compare}
    The lower bound derived in Theorem~\ref{prop:lbm2} is tighter than the one derived in Theorem~\ref{prop:genlb} for $\dbdim=2$ and for all $q \ge 3$. For $\dbdim=2$ and $q=2$ the two bounds are the same.
\end{observation}

To prove Theorem~\ref{prop:lbm2}, we require the following two claims.

\begin{restatable}{claim}{claimconstrtwotwo} \label{claim:constr2_2}
  For any walk of length 3 $(x_1, x_2, x_3, x_4, x_5)$ in~$\mG^2_q$ it holds that~$x_3=x_4$.
\end{restatable}
\begin{proof}
      Let $(x_1, x_2, x_3, x_4, x_5)$ be a walk of length 3 in $\mG^2_{q}$. We perform a case distinction based on the color of the last edge $(x_3, x_4, x_5)$.
  If $(x_3, x_4, x_5)$ is red or blue then $x_3=x_4$ holds. It remains to show that $(x_3, x_4, x_5)$ cannot be black, green or purple.
  \begin{itemize}
    \item \emph{Assume $(x_3, x_4, x_5)$ is \ecblack.} Since we have $x_3>x_4$ the edge $(x_2, x_3, x_4)$ can only be \ecpurple, i.e., $q-1=x_2>x_3>x_4>0$. However, then there exists no edge $(x_1, x_2, x_3)$.  
    \item \emph{Assume $(x_3, x_4, x_5)$ is \ecgreen.} We have $0=x_3<x_4\leq x_5$. Then, $(x_2, x_3, x_4)$ can only be \ecblack implying $x_2\neq 0$ and $x_3=0$ in which case there exists no edge $(x_1, x_2, x_3)$.
    \item \emph{Assume $(x_3, x_4, x_5)$ is \ecpurple.} If the edge $(x_3, x_4, x_5)$ is \ecpurple then $q-1=x_3>x_4>x_5>0$ and there exists no possible edge $(x_2, x_3, x_4)$. \qedhere
  \end{itemize}
\end{proof}

\begin{restatable}{claim}{claimconstrtwoone} \label{claim:constr2_1}
  For any $x_1, x_2, x_4, x_5 \in \Sigma_q$, there exists an $x_3 \in \Sigma_q$ such that $(x_1, x_2, x_3, x_4, x_5)$ is a walk of length~3 in $\mG^2_{q}$ if and only if $(x_1, x_2, x_4, x_5)$ is a walk in $\mG^2_{q}$.
  Furthermore, if $x_3$ exists, it is unique.
\end{restatable}
\begin{proof}
  Claim~\ref{claim:constr2_2} shows that for any walk of length 3 $(x_1, x_2, x_3, x_4, x_5)$ in $\mG^2_q$ it holds that $x_3=x_4$.
  From the claim it directly follows that for any walk of length 2 with the same start and end vertex $(x_1, x_2, x_4, x_5)$ there can be at most one $x_3$ such that $(x_1, x_2, x_3, x_4, x_5)$ is a walk of length 3 in $\mG^2_q$.

    We first show that if $(x_1, x_2, x_3, x_4, x_5)$ is a walk in $\mathcal{G}^{2}_q$ then so is $(x_1, x_2, x_4, x_5)$, i.e., if $(x_1, x_2, x_3)$, $(x_2, x_3, x_4)$, and $(x_3, x_4, x_5)$ are edges in $\mG^2_q$, then so are $(x_1, x_2, x_4)$ and $(x_2, x_4, x_5)$.

  Since $x_3=x_4$, if $(x_1, x_2, x_3)$ is an edge in $\mG^2_q$, then so is $(x_1, x_2, x_4)$. 
  We proceed to show that $(x_2, x_4, x_5) = (x_2, x_3, x_5)$ is an edge in $\mG^2_q$ by case distinction on the color of the edges $(x_3, x_3, x_5)$ and $(x_2, x_3, x_3)$.
  An edge of the form $(x_3, x_3, x_5)$ can be \ecblue or \ecred.
  \begin{itemize}
    \item If $(x_3, x_3, x_5)$ is \ecblue, then $x_3=x_4=x_5$ and $(x_2, x_4, x_5)$ exists and equals $(x_2, x_3, x_4)$.
    \item If $(x_3, x_3, x_5)$ is \ecred, then $x_3<x_5$ and $(x_2, x_3, x_3)$ can be \ecblue, \ecblack, or \ecgreen.
    \begin{itemize}
      \item If $(x_2, x_3, x_3)$ is \ecblue, then $x_2=x_3$ and $(x_2, x_4, x_5) = (x_3, x_4, x_5)$.
      \item If $(x_2, x_3, x_3)$ is \ecblack, then $x_2 > x_3 = x_4<x_5$ and $(x_2, x_4, x_5)$ is also \ecblack.
      \item If $(x_2, x_3, x_3)$ is \ecgreen, then $x_2<x_4<x_5$ and $(x_2, x_4, x_5)$ is also \ecgreen.
    \end{itemize}
  \end{itemize}

  It remains to show that if $(x_1, x_2, x_4)$ and $(x_2, x_4, x_5)$ are edges in $\mG^2_q$, then so are $(x_1, x_2, x_3)=(x_1, x_2, x_4)$, $(x_2, x_3, x_4)=(x_2, x_4, x_4)$, and $(x_3, x_4, x_5)=(x_4, x_4, x_5)$.
  Clearly, if $(x_1, x_2, x_3)$ exists, then $(x_1, x_2, x_4)$ is the same edge and also exists. We proceed by case distinction on the color of $(x_2, x_4, x_5)$.\begin{itemize}
      \item If $(x_2, x_4, x_5)$ is \ecblue, then $x_2=x_3=x_4=x_5$ and the edges $(x_2, x_3, x_4)$ and $(x_3, x_4, x_5)$ exist and are \ecblue.
      \item If $(x_2, x_4, x_5)$ is \ecred, then $0<x_2 = x_4 < x_5$ and $(x_2, x_3, x_4)$ is \ecblue and $(x_3, x_4, x_5)$ is \ecred.
      \item If $(x_2, x_4, x_5)$ is \ecblack, then $x_2 > x_4$ and $x_4 \leq x_5$. The edge $(x_2, x_3, x_4)$ is \ecblack and $(x_3, x_4, x_5)$ is \ecblue if $x_4 = x_5$ and \ecred if $x_4 < x_5$.
      \item If $(x_2, x_4, x_5)$ is \ecgreen, then $0=x_2<x_4\leq x_5$. The edge $(x_2, x_3, x_4)$ is also \ecgreen. If $x_4 = x_5$, $(x_3, x_4, x_5)$ is \ecblue and if $x_4 < x_5$ it is \ecred.
      \item The edge $(x_2, x_4, x_5)$ cannot be purple, since if it were, then $q-1 = x_2 > x_4$ and $(x_1, x_2, x_4)$ does not exist. \qedhere
  \end{itemize}
\end{proof}

\begin{proof}[Proof of Theorem~\ref{prop:lbm2}] 
  Let $\bm{A}$ be the adjacency matrix of $\mG^2_{q}$. From Claim~\ref{claim:constr2_1} it follows that $\bm{A}^2=\bm{A}^3$. {Induction on $i$ shows that} $\bm{A}^2=\bm{A}^i$ for all integers $i \ge 2$. Since $\bm{B}_{q,\dbdim}$ and $\bm{B}_{q,\dbdim}^2$ have no entry greater than 1, neither have $\bm{A}$ and $\bm{A}^2$.
\end{proof}

\section{An Upper Bound on $\gamma(q,\dbdim)$} \label{sec:upperbound}
Inspired by the upper bound on $\gamma(q, 1)$ in~\cite{zhan2012Extremal}, in this section we give a general upper bound on $\gamma(q,\dbdim)$. We will need the following notation.

\begin{notation}
For a positive integer $k$ we denote by $\eta(q,\dbdim,k)$ the maximum number of distinct walks in $\mB_{q,\dbdim}$ of length $k$ that have at least one edge in common. 
\end{notation}

\begin{theorem} \label{prop:genub}
  For any positive integer $k$, we have
  \begin{align*}
    \gamma(q, \dbdim) \leq q^{\dbdim+1} - \frac{q^{\dbdim+k} - \gamma(q^\dbdim,1)}{\eta(q, \dbdim, k)}.
  \end{align*}
\end{theorem}
\begin{proof}
  Let $\mG = (\mV_{q, \dbdim},\mE) \in \Gamma(q, \dbdim)$ be an arbitrary (but fixed) path unique subgraph of $\mB_{q,\dbdim} = (\mV_{q, \dbdim}, \mE_{q, \dbdim})$ with adjacency matrix $\bm{A} \in \{0, 1\}^{q^\dbdim\times q^\dbdim}$.
  By Lemma~\ref{lem:adj}, for any positive integers~$k,i$ it holds that $\bm{A}^k \in \{0, 1\}^{q^\dbdim \times q^\dbdim}$ and $\bm{A}^{ki} \in \{0, 1\}^{q^\dbdim \times q^\dbdim}$. The former means that $\bm{A}$ is an adjacency matrix of a graph and the latter implies that the graph is path unique. %
  By Theorem~\ref{thm:dunno}, we have $|\bm{A}^k| \leq \gamma(q^\dbdim, 1)$. Using this fact, we proceed to bound the number of edges in~$\mG$ from above.

 We define the set $\mS = \mE_{q, \dbdim} \setminus \mE $ to be the set of edges that are in $\mB_{q, \dbdim}$ but not in~$\mG$.
 The set $\mW$ of walks of length $k$ in $\mG$ is the set of walks of length $k$ in $\mB_{q,\dbdim}$ that do not traverse any edge in $\mS$. We have $|\mW| = |\bm{A}^k| \leq \gamma(q^\dbdim, 1)$ and $|\mS| = |\mE_{q, \dbdim}| - |\mE| = q^{\dbdim+1} - |\bm{A}|$.
Since any edge in $\mS$ is traversed by at most $\eta(q, \dbdim, k)$ distinct walks of length $k$, it holds that the number of walks in $\mW$ is larger than or equal to the total number of walks of length $k$, $q^{\dbdim+k}$, minus the number of walks that traverse edges in $\mS$, bounded by $\eta(q, \dbdim, k) |\mS|$. Therefore, we have $|\mW| \geq q^{\dbdim+k} - \eta(q, \dbdim, k) |\mS|$, which gives $|\mS| \geq \frac{q^{\dbdim+k}-|\mW|}{\eta(q, \dbdim, k)}$.
 Combining the above, we obtain
 \begin{align*}
    \gamma(q, \dbdim) &\leq |\bm{A}| 
                 = q^{\dbdim+1} - |S| 
                 \leq q^{\dbdim+1} - \frac{q^{\dbdim+k} - \gamma(q^\dbdim, 1)}{\eta(q, \dbdim, k)}
 \end{align*}
 concluding the proof.
\end{proof}

In order to explicitly evaluate the bound of Theorem~\ref{prop:genub}, we need Theorem~\ref{thm:extremal} and the following result.

\begin{restatable}{lemma}{lemmaincl} \label{lem:incl}
It holds that
\begin{align*}
    \eta(q,\dbdim,k) = \sum_{i=1}^{\left\lfloor \frac{\dbdim+k}{\dbdim+1}\right\rfloor} (-1)^{i+1}q^{\dbdim+k-i(\dbdim+1)}\binom{k-(i-1)\dbdim}{i}.
\end{align*}
\end{restatable}
\begin{proof}
  We use the correspondence between walks in $\mB_{q, \dbdim}$ and sequences over $\Sigma_q$.
  To determine $\eta(q, \dbdim, k)$, we count the maximum number of sequences $\bm{b} \in \Sigma_q^{\dbdim+k}$ that contain some $\bm{a} \in \Sigma_q^{\dbdim+1}$ as a contiguous subsequence, i.e., $\eta(q,\dbdim,k)$ is equal to
  \begin{align*}
    \max_{\bm{a} \in \Sigma_q^{\dbdim+1}}\!\!\big|\!\big\{\bm{b}\!\in\!\Sigma_q^{\dbdim+k}\!:\! \bm{b}_{[\ell,\dbdim+1]}\!=\!\bm{a} \text{ for some } \ell\!\in\![k]\big\}\!\big|.
  \end{align*}%
  In~\cite[Section 7]{guibas1981string} it is shown that the number of sequences not containing a string monotonically increases with the autocorrelation of the string. 
  Hence, the maximum is attained for a sequence $\bm{a}$ with the lowest possible autocorrelation, i.e., a string that does not overlap itself, e.g. $\bm{a} = (1, 0, 0, \dots, 0)$.
  For $\dbdim+k < 2\dbdim$ we have $\eta(q, \dbdim, k) = k q^{\dbdim+k-i(\dbdim+1)}$ since there are $k$ positions where $\bm{a}$ can be placed in $\bm{b}$ and $q^{\dbdim+k-(\dbdim+1)}$ possibilities for choosing the remaining symbols of $\bm{b}$.
  In the general case, to avoid overcounting the sequences~$\bm{b}$ that contain $\bm{a}$ more than once, we use the Inclusion-Exclusion Principle to obtain
\begin{align*} 
    \eta(q,\dbdim,k) = \sum_{i=1}^{\lfloor \frac{\dbdim+k}{\dbdim+1} \rfloor}(-1)^{i+1} q^{\dbdim+k-i(\dbdim+1)}P_i,
\end{align*}
where $P_i$ denotes the number of possible positions to place~$i$ copies of $\bm{a}$ into $\bm{b}$. It is multiplied by the number of ways to choose the remaining symbols of $\bm{b}$, which is $q^{\dbdim+k-i(\dbdim+1)}$. 
Note that~$P_i$ is the same as the number of integer vectors $(x_1,\dots,x_i)$ with $1 \le x_1 \le \dots \le x_i \le \dbdim+k-i(\dbdim+1)$ (each $x_j$ can be seen as the position for which the length-$(\dbdim+k)$ vector $\bm{b}$ satisfies $\bm{b}_{[(j-1)(\dbdim+1)+x_j,\dbdim+1]}=\bm{a}$, i.e., the position of $\bm{b}$ in which $\bm{a}$ starts). Similar to the proof of Lemma~\ref{lem:constr1count} (also here see e.g.~\cite[Chapter II.5]{feller1991introduction} for a reference), we have $P_i = \binom{\dbdim+k-i(\dbdim+1)+i}{i}$. This proves the lemma.
\end{proof}

\begin{corollary} \label{cor:upperbound}
    By setting $k=\dbdim$ and $k=\dbdim+1$, respectively, in Theorem~\ref{prop:genub} we get the following bounds%
    \begin{align*}
        \gamma(q, \dbdim) \leq q^{\dbdim+1} \Big(1 &-\frac{3}{4\dbdim} + \frac{1}{2\dbdim q^\dbdim}\Big), \\
        \gamma(q, \dbdim) \leq q^{\dbdim+1} \Big(1 &- \frac{1}{\dbdim+1} + \frac{1}{4q(\dbdim+1)} +\\
        & \frac{1}{2(\dbdim+1)q^{\dbdim+1}} +\frac{1}{4(\dbdim+1)q^{2\dbdim+1}}\Big).
    \end{align*}
    In particular, for every fixed $\dbdim$ it holds that $$\lim_{q\rightarrow \infty}\frac{\gamma(q,\dbdim)}{q^{\dbdim+1}} \le \min\left\{1-\frac{1}{\dbdim+1},1-\frac{3}{4\dbdim}\right\}.$$
\end{corollary}
For $\dbdim \to \infty$  the upper bound  on ${\gamma(q,\dbdim)}/{q^{\dbdim+1}}$ is 1.

\section{Connections to %
Universal Hitting Sets and Decycling Sets} \label{sec:connections}

In this section, we establish connections between Problem~\ref{prob:gamma} and existing problems in the literature, enabling us to prove that as $d$ tends to infinity, only a negligible fraction of the edges need to be removed from $\mB_{q,d}$ to make it path unique. Equivalently, only a negligible fraction of all length-$(d+1)$ labels are required to achieve full labeling capacity.  

In more detail, in this section, we highlight connections to \emph{universal hitting sets} and \emph{minimum decycling sets} of the de Bruijn graph.  
Through these connections and results from the literature, we prove the following theorem in two ways.

\begin{theorem} \label{thm:asymptotic}
  For all finite values of $q$, as $d$ tends to infinity, the maximum number of edges in a path unique %
  subgraph of the de Bruijn graph $\mathcal{B}_{q, d}$ tends to $q^{d+1}$, i.e.,
  \begin{align*}
    \lim_{d\to \infty} \frac{\gamma(q, d)}{q^{d+1}} = 1.
  \end{align*}
\end{theorem}

\subsection{Connections to Universal Hitting Sets}

We prove Theorem~\ref{thm:asymptotic} in the rest of this subsection by connecting Problem~\ref{prob:gamma} to the related problems of \emph{unavoidable sets of strings} (of constant length) \cite{champarnaud2004unavoidable}, and \emph{universal $k$-mer hitting sets} \cite{orenstein2016compact}.

\begin{notation} %
For natural numbers $m$ and $n$ with $m\leq n$, we say that a string $x \in \Sigma_q^n$ 
\emph{avoids} a set of strings $\mathcal{S} \subseteq \Sigma_q^m$, %
if no element of $\mathcal{S}$ is a contiguous substring of $x$. A set of strings is unavoidable by strings of length $n$ ($n$-unavoidable for short) if there does not exist a string $x\in \Sigma_q^n$ that avoids $\mathcal{S}$. 
\end{notation}

\begin{definition}[Minimum Universal $m$-mer Hitting Sets]
    Let $\Psi_{q, m, n} \subseteq \Sigma_q^m$ denote an $n$-unavoidable set of strings of length $m\leq n$ with least cardinality. Let $\psi(q, m, n)=|\Psi_{q, m, n}|$.
\end{definition}

We formally state the problem as follows.
\begin{problem} \label{prob:psi}
  Given positive integers $q, m$, and $n$, find $\psi(q, m, n)$ and the corresponding set $\Psi_{q, m, n}$.
\end{problem}

Universal $m$-mer Hitting Sets are related to path unique de Bruijn graphs by the following claim.

\begin{claim} \label{claim:relation_gamma_psi}
    It holds that $\gamma(q, d) \geq q^{d+1} - \psi(q, d+1, 2d+1)$.
\end{claim}
\begin{proof}
Consider the de Bruijn graph $\mathcal{B}_{q, d} = (\mathcal{V}_{q,d}, \mathcal{E}_{q,d})$. Construct the graph $\mathcal{G}_\Psi = (\mathcal{V}_{q,d}, \mathcal{E}_{q,d}\setminus \Psi_{q, d+1, 2d+1})$ on the same vertex set but with edges corresponding to strings in $\Psi_{q, d+1, 2d+1}$ removed. $\mathcal{G}_\Psi$ has $q^{d+1}-\psi(q, d+1, 2d+1)$ edges. It remains to show that $\mathcal{G}_\Psi$ is a path unique graph.
Walks corresponding to strings of length at most $2d$ are fully determined by the start and end vertex. Any walk corresponding to a string of length $2d+1$ in $\mathcal{B}_{q, d}$ is guaranteed to traverse an edge in $\Psi_{q, d+1, 2d+1}$ and is thus not in $\mathcal{G}_\Psi$.
\end{proof}

The values of $\psi(q, m, n)$ are of interest in the literature on \emph{minimizer sketches}, which refer to a certain way of deriving compact representations of long strings, for applications in computational biology. For a comprehensive overview we refer the reader to~\cite{zheng2022theory,zheng2023creating}. Next to a plethora of works \cite{marccais2017improving,orenstein2017designing,deblasio2019practical,ekim2020randomized,golan2024greedymini} on algorithms and heuristics for finding $n$-unavoidable sets of low cardinality, there are some theoretical bounds on $\psi(q, m, n)$. %
In \cite{marccais2018asymptotically}, it is demonstrated that for even $q$ and any fixed integer $w \geq 1$ it holds that $\lim_{m\to \infty} \frac{\psi(q, m, m+w-1)}{q^m} = \frac{1}{w}$. 
For any sufficiently large integer $m$, \cite{zheng2021lower} constructs a $(2m-1)$-unavoidable set of strings of length $m$ with small asymptotic size. We state this result~\cite[Theorem~2]{zheng2021lower} here for completeness.
\begin{theorem} \label{thm:forbiddenwordset}
    For sufficiently large integers $m$, it holds that 
    \begin{align*}
        \frac{\psi(q, m, 2m-1)}{q^m} \leq O\left(\frac{\ln(m)}{m}\right).
    \end{align*}
\end{theorem}

Using this result we can prove Theorem~\ref{thm:asymptotic}.

\begin{proof}[Proof of Theorem \ref{thm:asymptotic}]
    Combining Claim~\ref{claim:relation_gamma_psi} and Theorem~\ref{thm:forbiddenwordset}, we have %
    \begin{align*}
    \lim_{d\to\infty} \frac{\gamma(q, d)}{q^{d+1}} &\geq  1 - \lim_{d\to\infty} \frac{\psi(q, d+1, 2d+1)}{q^{d+1}} \\
                 &\geq 1 - \lim_{d\to\infty} O\left(\frac{\ln(d)}{d}\right) \\
                 &= 1. \qedhere
    \end{align*}%
\end{proof}

\subsection{Connections to Minimum Decycling Sets}

Minimum decycling sets of the de Bruijn graph are closely related to combinatorial \emph{necklaces}.

\begin{definition}[$q$-ary necklace]
A \textit{$q$-ary necklace of length~$m$} is a circular sequence of $m$ symbols from $\Sigma_q$. Two necklaces are said to be equivalent if one can be made equal to the other by rotation (without flipping it over). We will use the lexicographically smallest string in each equivalence class to represent the necklace.
\end{definition}

The number of $q$-ary necklaces of length $m$ is given by the necklace polynomial $\mathcal{N}(q, m) :=  \frac{1}{m} \sum_{i=1}^m q^{\mathrm{gcd}(i, m)}$, which is well approximated by $\mN(q, m) \approx {q^m}/{m}$ even for moderate $q$ and $m$.

The set of $q$-ary necklaces of length $m$ partitions the set of vertices of $\mB_{q, m}$ into $\mN(q, m)$ disjoint cycles. Every vertex is in exactly one of the cycles corresponding to the necklaces. Thus the minimum possible size of a decycling set of $\mB_{q, m}$ is at least $\mN(q, m)$. %

\begin{definition}[Minimum Decycling Set]
  A set of vertices $\mathcal{S} \in \Sigma_q^{m}$ in $\mB_{q, m}$ is called a minimum decycling set of $\mB_{q,m}$ if $|\mS| = \mN(q, m)$, and any cycle in $\mB_{q, m}$ visits at least one vertex in $\mS$. 
\end{definition}

A classic result by Mykkeltveit~\cite{mykkeltveit1972proof} shows that for any given $q$ and $m$, there exists a minimum decycling set, called the Mykkeltveit \emph{V-Set}. 
Equivalently, the Mykkeltveit V-set corresponds to a set of $\mN(q, m)$ strings of length $m$ that cannot be avoided  by infinitely long strings. It follows that $\psi(q, m, n) \geq \mathcal{N}(q, m)$ for all $n$. 

\begin{figure}[t]
  \centering
  \resizebox{0.4\textwidth}{!}{\definecolor{tableaublue}{HTML}{1F77B4}
\definecolor{tableauorange}{HTML}{FF7F0E}
\definecolor{tableaugreen}{HTML}{2CA02C}
\definecolor{tableaured}{HTML}{D62728}
\definecolor{tableaupurple}{HTML}{9467BD}
\definecolor{tableaubrown}{HTML}{8C564B}
\definecolor{tableaupink}{HTML}{E377C2}
\definecolor{tableaugrey}{HTML}{7F7F7F}
\definecolor{tableaulightgreen}{HTML}{BCBD22}
\definecolor{tableaucyan}{HTML}{17BECF}

\begin{tikzpicture}[->,>=stealth',shorten >=1pt,thick,round/.style={draw,circle,minimum size=5mm},]
\SetGraphUnit{2} 
\node[round, fill=white!40!tableauorange, dashed] at (-3,0) (A000) [] {${000}$};
\node[round, fill=white!40!tableaugreen, dashed] (A010) at (-1,0) [] {${010}$};
\node[round, fill=white!40!tableaublue] (A101) at (1,0) [] {$101$};
\node[round, fill=white!40!tableaupink, dashed] (A111) at (3,0) [] {${111}$};

\node[round, fill=white!40!tableaugreen] (A001) at (-3,2) [] {$001$};
\node[round, fill=white!40!tableaugreen] (A100) at (-3,-2) [] {$100$};

\node[round, fill=white!40!tableaublue] (A011) at (3,2) [] {$011$};
\node[round, fill=white!40!tableaublue, dashed] (A110) at (3,-2) [] {${110}$};

\Loop[dist=1cm,dir=EA,color=tableauorange](A000)
\Loop[dist=1cm,dir=WE,color=tableaupink](A111)
\draw[->,color=tableaugreen] (A100) to [bend left] node [above] {} (A001);
\draw[->] (A001) to [left] node [above] {} (A011);
\draw[->,color=tableaublue] (A011) to [bend left] node [above] {} (A110);
\draw[->] (A110) to [left] node [above] {} (A100);

\draw[->] (A100) to [left] node [above] {} (A000);
\draw[->] (A000) to [left] node [above] {} (A001);
\draw[->,color=tableaugreen] (A001) to [left] node [above] {} (A010);
\draw[->,color=tableaugreen] (A010) to [left] node [above] {} (A100);

\draw[->] (A010) to [bend left] node [above] {} (A101);
\draw[->] (A101) to [bend left] node [above] {} (A010);

\draw[->,color=tableaublue] (A101) to [left] node [above] {} (A011);
\draw[->,color=tableaublue] (A110) to [left] node [above] {} (A101);
\draw[->] (A011) to [left] node [above] {} (A111);
\draw[->] (A111) to [left] node [above] {} (A110);

\end{tikzpicture}}
\caption{\small The de Bruijn graph $\mB_{2, 3}$ with the cycles corresponding to necklaces highlighted by color; 000 in \textcolor{tableauorange}{\emph{orange}}, 001 in \textcolor{tableaugreen}{\emph{green}}, 011 in \textcolor{tableaublue}{\emph{blue}}, and 111 in \textcolor{tableaupink}{\emph{pink}}. %
The four nodes with dashed circles form a minimum decycling set of the graph.}
\label{fig:necklacepartition}
\end{figure}

\begin{figure*}[t]
  \begin{subfigure}{0.5\textwidth}
    \centering
  \resizebox{0.8\textwidth}{!}{
    \input{tikz/figures/extensionplot_1}
  }
  \vspace{-1em}
  \caption{\small Length of longest avoiding walk.}
  \label{fig:ext1}
\end{subfigure}
\hfill
\begin{subfigure}{0.5\textwidth}
  \centering
  \resizebox{0.8\textwidth}{!}{
    \begin{tikzpicture}[/tikz/background rectangle/.style={fill={rgb,1:red,1.0;green,1.0;blue,1.0}, fill opacity={1.0}, draw opacity={1.0}}, show background rectangle]
\begin{axis}[point meta max={nan}, point meta min={nan}, legend cell align={left}, legend columns={1}, title={}, title style={at={{(0.5,1)}}, anchor={south}, font={{\fontsize{14 pt}{18.2 pt}\selectfont}}, color={rgb,1:red,0.0;green,0.0;blue,0.0}, draw opacity={1.0}, rotate={0.0}, align={center}}, legend style={color={rgb,1:red,0.0;green,0.0;blue,0.0}, draw opacity={1.0}, line width={1}, solid, fill={rgb,1:red,1.0;green,1.0;blue,1.0}, fill opacity={1.0}, text opacity={1.0}, font={{\fontsize{8 pt}{10.4 pt}\selectfont}}, text={rgb,1:red,0.0;green,0.0;blue,0.0}, cells={anchor={center}}, at={(0.98, 0.98)}, anchor={north east}}, axis background/.style={fill={rgb,1:red,1.0;green,1.0;blue,1.0}, opacity={1.0}}, anchor={north west}, xshift={77.2mm}, yshift={-1.0mm}, 
  width={66mm}, height={66mm}, 
  scaled x ticks={false}, xlabel={$d$}, x tick style={color={rgb,1:red,0.0;green,0.0;blue,0.0}, opacity={1.0}}, x tick label style={color={rgb,1:red,0.0;green,0.0;blue,0.0}, opacity={1.0}, rotate={0}}, xlabel style={at={(ticklabel cs:0.5)}, anchor=near ticklabel, at={{(ticklabel cs:0.5)}}, anchor={near ticklabel}, font={{\fontsize{11 pt}{14.3 pt}\selectfont}}, color={rgb,1:red,0.0;green,0.0;blue,0.0}, draw opacity={1.0}, rotate={0.0}}, xmajorgrids={true}, xmin={0}, xmax={272}, xticklabels={{$0$,$50$,$100$,$150$,$200$,$250$}}, xtick={{0.0,50.0,100.0,150.0,200.0,250.0}}, xtick align={inside}, xticklabel style={font={{\fontsize{8 pt}{10.4 pt}\selectfont}}, color={rgb,1:red,0.0;green,0.0;blue,0.0}, draw opacity={1.0}, rotate={0.0}}, x grid style={color={rgb,1:red,0.8;green,0.8;blue,0.8}, draw opacity={1.0}, line width={0.5}, solid}, axis x line*={left}, x axis line style={color={rgb,1:red,0.0;green,0.0;blue,0.0}, draw opacity={1.0}, line width={1}, solid}, scaled y ticks={false}, ylabel={rel. cardinality}, y tick style={color={rgb,1:red,0.0;green,0.0;blue,0.0}, opacity={1.0}}, y tick label style={color={rgb,1:red,0.0;green,0.0;blue,0.0}, opacity={1.0}, rotate={0}}, ylabel style={at={(ticklabel cs:0.5)}, anchor=near ticklabel, at={{(ticklabel cs:0.5)}}, anchor={near ticklabel}, font={{\fontsize{11 pt}{14.3 pt}\selectfont}}, color={rgb,1:red,0.0;green,0.0;blue,0.0}, draw opacity={1.0}, rotate={0.0}}, ymajorgrids={true}, ymin={0}, ymax={0.6}, yticklabels={{$0.0$,$0.2$,$0.4$,$0.6$,$0.8$,$1.0$}}, ytick={{0.0,0.2,0.4,0.6000000000000001,0.8,1.0}}, ytick align={inside}, yticklabel style={font={{\fontsize{8 pt}{10.4 pt}\selectfont}}, color={rgb,1:red,0.0;green,0.0;blue,0.0}, draw opacity={1.0}, rotate={0.0}}, y grid style={color={rgb,1:red,0.8;green,0.8;blue,0.8}, draw opacity={1.0}, line width={0.5}, solid}, axis y line*={left}, y axis line style={color={rgb,1:red,0.0;green,0.0;blue,0.0}, draw opacity={1.0}, line width={1}, solid}, colorbar={false}]
    \addplot[color={rgb,1:red,0.1216;green,0.4667;blue,0.7059}, name path={299}, draw opacity={1.0}, line width={1}, solid, mark={*}, mark size={2.25 pt}, mark repeat={1}, mark options={color={rgb,1:red,0.0;green,0.0;blue,0.0}, draw opacity={1.0}, fill={rgb,1:red,0.1216;green,0.4667;blue,0.7059}, fill opacity={1.0}, line width={0.75}, rotate={0}, solid}]
        table[row sep={\\}]
        {
            \\
            3.0  0.5  \\
            4.0  0.375  \\
            5.0  0.25  \\
            6.0  0.21875  \\
            7.0  0.15625  \\
            8.0  0.140625  \\
            9.0  0.1171875  \\
            10.0  0.10546875  \\
            11.0  0.091796875  \\
            12.0  0.0859375  \\
            13.0  0.0771484375  \\
            14.0  0.0721435546875  \\
            15.0  0.06689453125  \\
            16.0  0.06280517578125  \\
        }
        ;
    \addlegendentry {V-set}
    \addplot[color={rgb,1:red,0.3392;green,0.5088;blue,0.1686}, name path={300}, draw opacity={1.0}, line width={1}, dotted, forget plot]
        table[row sep={\\}]
        {
            \\
            3.0  0.5  \\
            1.0  0.5  \\
        }
        ;
    \addplot[color={rgb,1:red,0.3392;green,0.5088;blue,0.1686}, name path={301}, draw opacity={1.0}, line width={1}, dotted, forget plot]
        table[row sep={\\}]
        {
            \\
            4.0  0.375  \\
            3.0  0.375  \\
        }
        ;
    \addplot[color={rgb,1:red,0.3392;green,0.5088;blue,0.1686}, name path={302}, draw opacity={1.0}, line width={1}, dotted, forget plot]
        table[row sep={\\}]
        {
            \\
            5.0  0.25  \\
            9.0  0.25  \\
        }
        ;
    \addplot[color={rgb,1:red,0.3392;green,0.5088;blue,0.1686}, name path={303}, draw opacity={1.0}, line width={1}, dotted, forget plot]
        table[row sep={\\}]
        {
            \\
            6.0  0.21875  \\
            16.0  0.21875  \\
        }
        ;
    \addplot[color={rgb,1:red,0.3392;green,0.5088;blue,0.1686}, name path={304}, draw opacity={1.0}, line width={1}, dotted, forget plot]
        table[row sep={\\}]
        {
            \\
            7.0  0.15625  \\
            25.0  0.15625  \\
        }
        ;
    \addplot[color={rgb,1:red,0.3392;green,0.5088;blue,0.1686}, name path={305}, draw opacity={1.0}, line width={1}, dotted, forget plot]
        table[row sep={\\}]
        {
            \\
            8.0  0.140625  \\
            33.0  0.140625  \\
        }
        ;
    \addplot[color={rgb,1:red,0.3392;green,0.5088;blue,0.1686}, name path={306}, draw opacity={1.0}, line width={1}, dotted, forget plot]
        table[row sep={\\}]
        {
            \\
            9.0  0.1171875  \\
            53.0  0.1171875  \\
        }
        ;
    \addplot[color={rgb,1:red,0.3392;green,0.5088;blue,0.1686}, name path={307}, draw opacity={1.0}, line width={1}, dotted, forget plot]
        table[row sep={\\}]
        {
            \\
            10.0  0.10546875  \\
            67.0  0.10546875  \\
        }
        ;
    \addplot[color={rgb,1:red,0.3392;green,0.5088;blue,0.1686}, name path={308}, draw opacity={1.0}, line width={1}, dotted, forget plot]
        table[row sep={\\}]
        {
            \\
            11.0  0.091796875  \\
            87.0  0.091796875  \\
        }
        ;
    \addplot[color={rgb,1:red,0.3392;green,0.5088;blue,0.1686}, name path={309}, draw opacity={1.0}, line width={1}, dotted, forget plot]
        table[row sep={\\}]
        {
            \\
            12.0  0.0859375  \\
            119.0  0.0859375  \\
        }
        ;
    \addplot[color={rgb,1:red,0.3392;green,0.5088;blue,0.1686}, name path={310}, draw opacity={1.0}, line width={1}, dotted, forget plot]
        table[row sep={\\}]
        {
            \\
            13.0  0.0771484375  \\
            141.0  0.0771484375  \\
        }
        ;
    \addplot[color={rgb,1:red,0.3392;green,0.5088;blue,0.1686}, name path={311}, draw opacity={1.0}, line width={1}, dotted, forget plot]
        table[row sep={\\}]
        {
            \\
            14.0  0.0721435546875  \\
            196.0  0.0721435546875  \\
        }
        ;
    \addplot[color={rgb,1:red,0.3392;green,0.5088;blue,0.1686}, name path={312}, draw opacity={1.0}, line width={1}, dotted, forget plot]
        table[row sep={\\}]
        {
            \\
            15.0  0.06689453125  \\
            217.0  0.06689453125  \\
        }
        ;
    \addplot[color={rgb,1:red,0.3392;green,0.5088;blue,0.1686}, name path={313}, draw opacity={1.0}, line width={1}, dotted, forget plot]
        table[row sep={\\}]
        {
            \\
            16.0  0.06280517578125  \\
            262.0  0.06280517578125  \\
        }
        ;
    \addplot[color={rgb,1:red,0.3392;green,0.5088;blue,0.1686}, name path={314}, const plot, draw opacity={1.0}, line width={1}, solid, mark={*}, mark size={2.25 pt}, mark repeat={1}, mark options={color={rgb,1:red,0.0;green,0.0;blue,0.0}, draw opacity={1.0}, fill={rgb,1:red,0.3392;green,0.5088;blue,0.1686}, fill opacity={1.0}, line width={0.75}, rotate={0}, solid}]
        table[row sep={\\}]
        {
            \\
            1.0  0.5  \\
            3.0  0.375  \\
            9.0  0.25  \\
            16.0  0.21875  \\
            25.0  0.15625  \\
            33.0  0.140625  \\
            53.0  0.1171875  \\
            67.0  0.10546875  \\
            87.0  0.091796875  \\
            119.0  0.0859375  \\
            141.0  0.0771484375  \\
            196.0  0.0721435546875  \\
            217.0  0.06689453125  \\
            262.0  0.06280517578125  \\
        }
        ;
    \addlegendentry {extended}
\end{axis}
\end{tikzpicture}
  }
  \vspace{-1em}
  \caption{\small Relative cardinality.}
  \label{fig:ext2}
\end{subfigure}
  \caption{\small Visualization of the Mykkeltveit V-set and its extension for the binary alphabet. The V-sets for $2\leq d \leq 16$ are extended until the length of the longest avoiding walk is $d+1$ (left figure). This increases $d$ but preserves the relative cardinality of the set (right figure). Since the relative cardinality of the V-set shrinks to zero as $d\to \infty$, so does the relative cardinality of the extended set.}
  \label{fig:extensionplot}
\end{figure*}

\begin{example}
  Let $q=2$ and $m=3$. Since $\mN(2, 3)=4$, there are four binary necklaces of length $3$, namely $000$, $001$, $011$, and $111$. Any other binary string of length $3$ is a cyclic shift of one of these four. Figure~\ref{fig:necklacepartition} shows how the necklaces partition the vertices of $\mB_{2, 3}$ into disjoint cycles. The Mykkeltveit V-Set for this example is $\{000, 010, 110, 111\}$. No walk of length $3$ avoids this set.  
\end{example}

Let $M_q(m)$ denote the length of the longest path avoiding the \emph{V-set}.
It is shown in \cite{zheng2021lower} that $M_q(m)$ is at least $O(m^2)$ and at most $O(m^3)$.
A construction for a different minimum decycling set of $\mB_{q, m}$ is given in \cite{champarnaud2004unavoidable}. %
In \cite{marccais2024sketching} a set of operations is defined that transform minimum decycling sets into other minimum decycling sets. While the cardinality of the different decycling sets remains the same, the length of the longest avoiding path in the resulting graph differs. %
Interestingly, it is shown in~\cite{marccais2024sketching} that the minimum decycling set constructed according to Mykkeltveit~\cite{mykkeltveit1972proof} seems to have relatively short longest avoiding paths compared to other minimum decycling sets.

To build a connection between decycling sets and $\gamma(q, m)$,  we define a useful operation on de Bruin graphs. It maps a subset of the vertices in $\mB_{q,d}$ to a subset of the vertices in $\mB_{q,d+1}$, while preserving the relative cardinality of the subset with respect to the total number of vertices in the respective graph as well as the length of longest avoiding walk. This operation is also used in \cite{deblasio2019practical} where it is called the \emph{naive extension}.

\begin{definition} \label{def:extend}
  For a set of strings $\mathcal{S}\subseteq \Sigma_q^m$ of length $m$ and for $\kappa \ge 0$, let $$\extend^\kappa(\mathcal{S}) := \{(\bm{i},\bm{x}): \bm{i}\in \Sigma_q^\kappa, x \in \mathcal{S}\}$$ denote the set of strings of length $m+\kappa$ with a string from $\mathcal{S}$ as a suffix, where we let $\extend^0(\mathcal{S}) := \mathcal{S}$.
\end{definition}

The following claim is adapted from \cite[Theorem 1]{deblasio2019practical}. We state it in our notation and provide a proof for completeness.

\begin{claim} \label{claim:extend}
  If $\mathcal{S}$ is an $n$-unavoidable set of strings of length~$m$, then $\extend^\kappa(\mathcal{S})$ is an ($n+\kappa$)-unavoidable set of strings of length $m+\kappa$. Further, $|\extend^\kappa(\mathcal{S})| = q^k |\mathcal{S}|$.
\end{claim}
\begin{proof}
  We prove the claim by induction on $\kappa$. For $\kappa=0$, the statement holds trivially. Assume it holds for $\kappa-1$. Consider the de Bruijn graph $\mathcal{B}_{q, m+\kappa-1}$. Elements of $\extend^{\kappa-1}(\mathcal{S})$ are  in $\mathcal{B}_{q, m+\kappa-1}$, elements of $\extend^{\kappa}(\mathcal{S})$ correspond to incoming edges of these vertices. Consider a walk $w$ corresponding to a string of length $n+\kappa$, i.e. a walk that visits $n-m+2$ vertices and traverses $n-m+1$ edges.
There are two walks corresponding to a string of length $n+\kappa-1$ contained in $w$.
One starting at the first vertex of $w$, called $w^\prime_1$ and one starting at the second vertex of $w$ called $w^\prime_2$. By the induction hypothesis $w^\prime_2$ contains a vertex $v$ that is in $\extend^{\kappa-1}(\mathcal{S})$.
Since all incoming edges of $\extend^{\kappa-1}(\mathcal{S})$ are in $\extend^{\kappa}(\mathcal{S})$, and $w$ traverses an incoming edge of $v$, it traverses at least one edge from $\extend^{\kappa}(\mathcal{S})$. 
  Thus, any string of length $n+\kappa$ has a substring in $\extend^{\kappa}(\mathcal{S})$.

  The cardinality of $\extend^\kappa(\mathcal{S})$ follows straightforwardly from Definition \ref{def:extend}.
\end{proof}

In other words, $\extend^\kappa$ maps a set of  $\mathcal{S}$ in the de Bruijn graph $\mB_{q, d}$ to a set of vertices in $\mB_{q,d+\kappa}$, while preserving the relative number of vertices $\frac{\mathcal{S}}{q^m} = \frac{\extend^\kappa(\mathcal{S})}{q^{m+\kappa}}$ and without increasing the number of vertices and edges in the respective longest avoiding walks.

\begin{proof}[Proof of Theorem~\ref{thm:asymptotic}]
  We establish that for any integer $q\geq2$ and real number $\varepsilon>0$, there exists an $m_\varepsilon$, such that for $m\geq m_\varepsilon$, it holds that\begin{align}
    \frac{\psi(q, m, 2m-1)}{q^m} < \varepsilon. \label{eq:proofasymptoticclaim}
  \end{align}
  The theorem then follows from Claim \ref{claim:relation_gamma_psi}.

  We will use the fact that for fixed $q$ and $m$ the function $\psi(q, m, n)$ is non-increasing in $n$ and that for any set $\mathcal{S}$, that is unavoidable by strings of length $n$, it holds that $\psi(q, m, n) \leq |\mathcal{S}|$. Specifically, for the Mykkeltveit V-Set we get $\psi(q, m, M_q(m)+1) = \mathcal{N}(q, m)$.

  For any $q, m$ and $\kappa\geq 1$ it holds that \begin{align*}
    &\frac{\psi(q, m+\kappa, 2(m+\kappa)-1)}{q^{m+\kappa}}  \\
    &\quad \quad \quad \leq \frac{\psi(q, m+\kappa, 2m+\kappa-1)}{q^{m+\kappa}}  \\
    &\quad \quad \quad \leq \frac{\psi(q, m, 2m-1)}{q^{m}},
  \end{align*}
  where the first step follows because $\psi(q, m, n)$ is non-increasing in $n$ and the second step follows from Claim~\ref{claim:extend} by considering the $\kappa$-fold extension of $\Psi(q, m, 2m-1)$.
  Thus, if \eqref{eq:proofasymptoticclaim} holds for $m=m_\varepsilon$, it holds for all $m \geq m_\varepsilon$.

\begin{figure*}[t]
  \begin{subfigure}{0.5\textwidth}
  \resizebox{\columnwidth}{!}{
    \input{tikz/figures/relative_number_of_edges_over_d_journal.tex}
  }
  \vspace{-1em}
  \caption{\small Relative number of edges over $d$.}
  \label{fig:q_2_4}
\end{subfigure}
\hfill
\begin{subfigure}{0.5\textwidth}
  \resizebox{\columnwidth}{!}{
    \input{tikz/figures/relative_number_of_edges_over_q_journal.tex}
  }
  \vspace{-1em}
  \caption{\small Relative number of edges over q.}
  \label{fig:m_2}
\end{subfigure}
\caption{\small Comparison of the upper bound (Thm.~\ref{prop:genub}) and the lower bounds (Thm.~\ref{prop:genlb} and Thm.~\ref{prop:lbm2}) on $\gamma(q, d)$. Note that Thm.~\ref{prop:lbm2} only applies for $d=2$. Constructions found by computer search for small $q$ and $d$ are shown for comparison (greedy, BLP, perm.).}
\label{fig:comparison}
\end{figure*}
  Let $m^\prime$ be the smallest positive integer, such that $\frac{\mN(q, m^\prime)}{q^{m^\prime}} < \varepsilon$. Such an integer exists since $\lim_{m\to \infty} \frac{\mathcal{N}(q, m)}{q^m} = 0$. 
  For notational convenience let $n^\prime := M_q(m^\prime)+1$.
  If $n^\prime \leq 2m^\prime-1$, we have
  \begin{align*}
      \frac{\psi(q, m^\prime, 2m^\prime-1)}{q^{m^\prime}} \leq \frac{\psi(q, m^\prime, n^\prime)}{q^{m^\prime}} < \varepsilon,
  \end{align*} and
  \eqref{eq:proofasymptoticclaim} holds with $m_\varepsilon = m^\prime$.

  If not, we consider the $(\kappa=n^\prime-2m^\prime + 1)$-fold extension of the Mykkeltveit V-set for $q=q$ and $m=m^\prime$ and let $m_\varepsilon = m^\prime + \kappa$.
  We have\begin{align*}
      & \frac{\psi(q, m_\varepsilon, 2m_\varepsilon-1)}{q^{m_\varepsilon}} \\
      & \quad \quad \quad = \frac{\psi(q, n^\prime - m^\prime + 1, 2n^\prime - 2m^\prime + 1)}{q^{n^\prime - m^\prime + 1}} \\
      & \quad \quad \quad = \frac{\psi(q, m^\prime + \kappa, n^\prime + \kappa)}{q^{m^\prime+ \kappa}} \\
     &\quad \quad \quad \leq \frac{\psi(q, m^\prime, n^\prime)}{q^{m^\prime}} \\
     &\quad \quad \quad < \varepsilon.
  \end{align*}
  Thus, in this case \eqref{eq:proofasymptoticclaim} holds with $m_\varepsilon = n^\prime -m^\prime + 1$.
\end{proof}

The walk length (in number of edges) and relative cardinality of this extension of the Mykkeltveit set are depicted in Figure~\ref{fig:extensionplot} for $q=2$. The cardinality is normalized by the total number of strings of the respective length.

\section{Numerical Comparisons} \label{sec:comparison}

In this section, we compare the lower bounds derived in Theorems~\ref{prop:genlb}~and~\ref{prop:lbm2} and the upper bound derived in Theorem~\ref{prop:genub}. 
Additionally, we include results obtained by computer search for small values of~$q$ and $\dbdim$.

Fig.~\ref{fig:q_2_4} displays the cases $2\leq q \leq 4$ over $\dbdim$.
For large~$d$, the lower bound in Theorem~\ref{prop:genlb} tends to ${1}/{2} + {1}/{2q}$ while the upper bound in Theorem~\ref{prop:genub} tends to $1$. 
From Theorem~\ref{thm:asymptotic}, it is clear that the upper bound is asymptotically tight. Fig.~\ref{fig:m_2} compares the results for $2\leq\dbdim\leq4$ over $q$. Theorem~\ref{prop:lbm2} gives a tighter bound than Theorem~\ref{prop:genlb} for finite values. Asymptotically in $q$, both lower bounds converge to $1/3$, while the upper bound tends to $0.625$. Additionally, results from computer search are included as elaborated next.

We conduct two types of computer search for path unique subgraphs of de Bruijn graphs $\mB_{q, \dbdim}$ for small values of $q$ and $\dbdim$. The first is loosely based on~\cite[Lemma~5]{zhan2012Extremal}. We search for permutation matrices $\bm{\Pi}$ such that the upper triangle of the permuted adjacency matrix $\mathrm{triu}(\bm{\Pi}^{-1} \bm{B}_{q, \dbdim} \bm{\Pi})$ is the adjacency matrix of a path unique graph and contains a high number of 1's. 
The search algorithm is in the spirit of %
simulated annealing, cf.~\cite[Chapter~30]{mackay2003information}. Starting with a random permutation, at each iteration a random swap is performed on the permutation matrix. The new permutation matrix is accepted if it results in a path unique matrix with a large number of 1's. If it does not, it is still accepted with a small probability; otherwise rejected. 

The second type of computer search is based on Claim~\ref{claim:relation_gamma_psi}. The problem of finding $\psi(q, d+1, 2d+1)$ is formulated as a minimum set cover problem, which is then either solved exactly as a binary linear program (with the mixed integer solver Gurobi) or approximately through greedy search. 
Formally, the problem can be written as
\begin{align}
  \min_{x \in \{0, 1\}^{q^{d+1}}} \sum_i x_i \text{ s.t. } BLP_{q, d} \cdot x \geq 1, \label{eq:BLP}
\end{align}
where the constraint matrix is given by
\begin{align*}
    BLP_{q, d} := \bigvee_{i=0}^d \boldsymbol{1}_{q^{d-i}} \otimes \boldsymbol{I}_{q^{d+1}} \otimes \boldsymbol{1}_{q^{i}},
\end{align*}
and $\boldsymbol{1}_x$ denotes the all-one column vector of $x$ elements, $\vee$ denotes the element-wise or operation, and $\otimes$ is the Kronecker product. 
The structure of $BLP_{q, d}$ arises naturally from arranging strings in lexicographic order.
The $(i,j)$th entry of the matrix $BLP_{q, d}$ equals $1$ whenever the lexicographically $j$th string of length $d+1$ is a substring of the lexicographically $i$th string of length $2d+1$.

The results of the computer search are reported in Table~\ref{tab:evaluations} along with numerical evaluations of our theoretical results.
For $q=2$ with $d\in\{2, 3, 4\}$ and $q=3$ with $\dbdim=2$ direct exhaustive search for path unique allocation matrices is feasible. 
The results for these values coincide with the search for permutations as reported in the table.

\begin{table}[t]
    \rowcolors{3}{gray!20}{white}
    \centering
    \renewcommand\arraystretch{1.1}
    \begin{tabularx}{0.49\textwidth}{ X c c c c c c c } 
      \hiderowcolors
  \multirow{2}{*}{$q$} & \multirow{2}{*}{$\dbdim$} & \multicolumn{2}{c}{LBs} & UB & \multicolumn{3}{c}{computer search (LBs)} \\
      & & Thm.~\ref{prop:genlb} & Thm.~\ref{prop:lbm2} &  Thm.~\ref{prop:genub} & BLP & greedy & perm.
  \\\noalign{\global\arrayrulewidth 1.8pt}
    \hline
    \noalign{\global\arrayrulewidth0.4pt}
      \showrowcolors
2  &  2  &  5$^*$  &  4  &  5$^*$  &  3  &  3  &  5$^*$ \\ \hline
2  &  3  &  11$^*$  &   -   &  12  &  9  &  8  &  11$^*$ \\ \hline
2  &  4  &  23  &   -   &  26  &  20  &  19  &  24$^*$ \\ \hline
2  &  5  &  47  &   -   &  54  &  45  &  44  &  50 \\ \hline
2  &  6  &  95  &   -   &  112  &  94  &  93  &  102 \\ \hline
2  &  7  &  191  &   -   &  228  &  198  &  193  &  210 \\ \hline
2  &  8  &  383  &   -   &  462  &   -   &  404  &  428 \\ \hline
2  &  9  &  767  &   -   &  934  &   -   &  817  &  864 \\ \hline
2  &  10  &  1535  &   -   &  1885  &   -   &  1670  & - \\ \hline
2  &  11  &  3071  &   -   &  3797  &   -   &  3370  & - \\ \hline
2  &  12  &  6143  &   -   &  7640  &   -   &  6878  & - \\ \hline
3  &  2  &  14  &  15$^*$  &  17  &  12  &  11  &  15$^*$ \\ \hline
3  &  3  &  49  &   -   &  61  &  46  &  44  &  53 \\ \hline
3  &  4  &  156  &   -   &  197  &   -   &  153  &  174 \\ \hline
3  &  5  &  479  &   -   &  617  &   -   &  499  &  553 \\ \hline
3  &  6  &  1450  &   -   &  1900  &   -   &  1559  &  1712 \\ \hline
3  &  7  &  4365  &   -   &  5809  &   -   &  4924  & - \\ \hline
4  &  2  &  30  &  33  &  41  &  28  &  26  &  34 \\ \hline
4  &  3  &  145  &   -   &  192  &   -   &  142  &  164 \\ \hline
4  &  4  &  619  &   -   &  832  &   -   &  638  &  680 \\ \hline
4  &  5  &  2532  &   -   &  3456  &   -   &  2797  &  2969 \\ \hline
5  &  2  &  55  &  63  &  79  &  55  &  52  &  64 \\ \hline
5  &  3  &  340  &   -   &  469  &   -   &  347  &  372 \\ \hline
5  &  4  &  1819  &   -   &  2531  &   -   &  1933  & - \\ \hline
6  &  2  &  91  &  107  &  136  &  96  &  91  &   - \\ \hline
6  &  3  &  686  &   -   &  973  &   -   &  720  &   - \\ \hline
6  &  4  &  4410  &   -   &  6285  &   -   &  4881  &   - \\ \hline
7  &  2  &  140  &  165  &  216  &   -   &  147  &   - \\ \hline
7  &  3  &  1246  &   -   &  1801  &   -   &  1337  &   - \\ \hline
8  &  2  &  204  &  240  &  322  &   -   &  216  &   - \\ \hline
8  &  3  &  2094  &   -   &  3073  &   -   &  2262  &   - \\ \hline
9  &  2  &  285  &  334  &  457  &   -   &  308  &   - \\ \hline
9  &  3  &  3315  &   -   &  4922  &   -   &  3676  &   - \\ \hline
10  &  2  &  385  &  449  &  627  &   -   &  417  &   - \\ \hline
10  &  3  &  5005  &   -   &  7501  &   -   &  5571  &   - \\ \hline
\end{tabularx}
\caption{\small Evaluations of the lower bounds (LBs) and the upper bound (UB) on $\gamma(q, d)$, derived in Theorems~\ref{prop:genlb},~\ref{prop:lbm2},~and~\ref{prop:genub} for small values of $q$ and~$\dbdim$.
Additionally, the results of the computer search are shown, where available. Exact and greedy solutions to \eqref{eq:BLP} are shown as ``BLP'' and ``greedy'', respectively, while ``perm.'' denotes the computer search for permutation matrices. An asterisk indicates that a result matches exhaustive search and is thus optimal.} %
\label{tab:evaluations}
\end{table}

Fig.~\ref{fig:q_inf} displays the asymptotic results for large alphabet sizes~$q$. For $\dbdim=2$, both lower bounds tend to ${1}/{3}$, while the upper bound from Corollary~\ref{cor:upperbound} tends to ${5}/{8}$.
For large $\dbdim$, the lower bound from Theorem~\ref{prop:genlb} tends to ${1}/{2}$ while the upper bound tends to $1$.

\begin{figure}[t]
  \centering
  \resizebox{0.4\textwidth}{!}{ \begin{tikzpicture}
\begin{axis}[
 width={100mm},
 height={60mm},
 xlabel={$\dbdim$},
 xmajorgrids={true},
 xticklabels={{$2$,$3$,$4$,$5$,$6$,$7$,$8$,$9$,$10$}},
 xtick={{2.0,3.0,4.0,5.0,6.0,7.0,8.0,9.0,10.0}},
 ylabel={$\lim_{q\to\infty}\gamma(q, \dbdim)/q^{\dbdim+1}$},
 ymajorgrids={true},
 ymin={0},
 ymax={1},
 yticklabels={{$0.0$,$0.1$,$0.2$,$0.3$,$0.4$,$0.5$,$0.6$,$0.7$,$0.8$,$0.9$,$1.0$}},
 ytick={{0.0,0.1,0.2, 0.3, 0.4,0.5, 0.6,0.7,0.8,0.9,1.0}},
 reverse legend,
legend style={at={(0.98,0.02)},anchor=south east}
 ]
    \addplot[color={rgb,1:red,0.9882;green,0.4902;blue,0.0431}, name path={64ed2383-7331-45b4-9b10-f3c7e7aa79cf}, draw opacity={1.0}, line width={1}, solid, mark={*}, mark size={3.0 pt}, mark repeat={1}, mark options={color={rgb,1:red,0.0;green,0.0;blue,0.0}, draw opacity={1.0}, fill={rgb,1:red,0.9882;green,0.4902;blue,0.0431}, fill opacity={1.0}, line width={0.75}, rotate={0}, solid}]
        table[row sep={\\}]
        {
            \\
            2.0  0.3333333333333333  \\
            3.0  0.4583333333333333  \\
            4.0  0.49166666666666664  \\
            5.0  0.4986111111111111  \\
            6.0  0.4998015873015873  \\
            7.0  0.4999751984126984  \\
            8.0  0.4999972442680776  \\
            9.0  0.49999972442680773  \\
            10.0  0.4999999749478916  \\
        }
        ;
    \addlegendentry {LB Thm. 3}
        \addplot[color={rgb,1:red,0.6392;green,0.6745;blue,0.7255}, name path={097c6c6a-cd27-492d-92ff-84e82f1fecdd}, draw opacity={1.0}, line width={1}, solid, mark={triangle*}, mark size={3.5 pt}, mark repeat={1}, mark options={color={rgb,1:red,0.0;green,0.0;blue,0.0}, draw opacity={1.0}, fill={rgb,1:red,0.6392;green,0.6745;blue,0.7255}, fill opacity={1.0}, line width={0.75}, rotate={0}, solid}]
        table[row sep={\\}]
        {
            \\
            2.0  0.3333333333333333  \\
        }
        ;
    \addlegendentry {LB Thm. 4}
\addplot[color={rgb,1:red,0.0667;green,0.4392;blue,0.6667}, name path={171deb0e-aeb1-4913-a157-d5cd5bd0dab5}, draw opacity={1.0}, line width={1}, solid, mark={square*}, mark size={3.0 pt}, mark repeat={1}, mark options={color={rgb,1:red,0.0;green,0.0;blue,0.0}, draw opacity={1.0}, fill={rgb,1:red,0.0667;green,0.4392;blue,0.6667}, fill opacity={1.0}, line width={0.75}, rotate={0}, solid}]
        table[row sep={\\}]
        {
            \\
            2.0  0.625  \\
            3.0  0.75  \\
            4.0  0.8  \\
            5.0  0.8333333333333334  \\
            6.0  0.8571428571428571  \\
            7.0  0.875  \\
            8.0  0.8888888888888888  \\
            9.0  0.9  \\
            10.0  0.9090909090909091  \\
        }
        ;
    \addlegendentry {UB Cor.~~3}

\end{axis}
\end{tikzpicture}}
\caption{\small Upper and lower bounds on the relative number of edges in an optimal path unique subgraph of the de Bruijn graph as the alphabet size tends to infinity.}
  \label{fig:q_inf}
\end{figure}

\section{Conclusion} \label{sec:conclusion}

Generalizing a problem initially presented in~\cite{hanania2023capacity}, we analyzed the value of $\gamma(q, \dbdim)$, which is the maximum number of edges in a path unique subgraph of the de Bruijn graph $\mB_{q, \dbdim}$.  We presented two constructions. The first provides a lower bound on $\gamma(q, \dbdim)$ for any $q$ and $\dbdim$, while the second is restricted to $\dbdim=2$. Additionally, we proved an upper bound on $\gamma(q, \dbdim)$. By drawing connections to other problems from the literature we showed that $\gamma(q, \dbdim)/q^{d+1}$ tends to $1$ as $\dbdim \to \infty$.

The results in this paper provide a significant step towards finding the values of $\gamma(q,\dbdim)$. 
Nevertheless, there remains an interesting gap between our lower and upper bounds for the vast majority of parameters. More progress on closing this gap, thus determining the exact values of $\gamma(q, \dbdim)$, remains for future work.

\bibliographystyle{IEEEtran}
\bibliography{main,refs_HananiaJournal}

\end{document}